\documentclass[%
 reprint,
 superscriptaddress,
 amsmath,amssymb,
 aps,
 prx,
]{revtex4-2}

\usepackage{amsthm}
\usepackage{graphicx}
\usepackage{hyperref}
\usepackage{cleveref}
\usepackage{xcolor}
\usepackage{dsfont}

\newtheorem{thm}{Theorem}
\newtheorem{lem}[thm]{Lemma}
\newtheorem{cor}[thm]{Corollary}
\newtheorem{prop}[thm]{Proposition}
\theoremstyle{definition}
\newtheorem{defn}[thm]{Definition}
\theoremstyle{remark}

\crefname{thm}{Theorem}{Theorems}
\crefname{lem}{Lemma}{Lemmas}
\crefname{cor}{Corollary}{Corollaries}
\crefname{prop}{Proposition}{Propositions}
\crefname{defn}{Definition}{Definitions}
\crefname{equation}{Eq.}{Eqs.}
\crefname{section}{Sec.}{Secs.}

\newcommand{\Phizeros}{|\Phi_0(s)\rangle}
\newcommand{\Phijs}{|\Phi_j(s)\rangle}
\newcommand{\Dstar}{\Delta_*}
\newcommand{\mubd}{\mu_<}
\newcommand{\Ran}{\operatorname{Ran}}
\newcommand{\HS}{\mathrm{HS}}
\DeclareMathOperator{\arsinh}{arsinh}

\begin{document}

\title{Bounds on adiabatic path geometry from the width class of the gap profile}

\author{Mancheon Han}
\email{mchan@kias.re.kr}
\affiliation{School of Computational Sciences, Korea Institute for Advanced Study (KIAS), Seoul, 02455, Korea}
\author{Sangkook Choi}
\email{sangkookchoi@kias.re.kr}
\affiliation{School of Computational Sciences, Korea Institute for Advanced Study (KIAS), Seoul, 02455, Korea}
\date{\today}

\begin{abstract}
  In an adiabatic evolution, the wave function follows the path of ground states $\Phizeros$ of a Hamiltonian
  $H(s)$ as $s$ is tuned from $0$ to $1$. The geometry of this adiabatic path is described by quantities such as
  its length $L=\int_0^1 \|\partial_s \Phizeros\|\,ds$ and its total curvature $K$. This geometry affects
  the evolution time $T$. For instance, $T$ is at least of order $L/\Delta_*$, where $\Delta_*$ is the minimum energy gap
  over $s\in[0,1]$. The traditional approach gives $L=\mathcal{O}(\Delta_*^{-1/2})$ and $K=\mathcal{O}(\Delta_*^{-1})$, which
  are loose in many cases. We improve these bounds by using $\mubd(\gamma)$, the width of the region in $s$ on
  which the gap is below $\gamma$. The gap profile is in width class $p$ when $\mubd(\gamma)$ falls at least as
  fast as $\gamma^{1/p}$. A typical avoided crossing is in width class $p=1$, which gives
  $L=\mathcal{O}\big(\sqrt{\log \Delta_*^{-1}}\big)$ and $K=\mathcal{O}(\log \Delta_*^{-1})$. A profile in
  width class $p>1$ gives power laws with exponents $(p-1)/(2p)$ for $L$ and $(p-1)/p$ for $K$. The width class also bounds the evolution time. At $p=1$ we have
  $T=\mathcal{O}(\Delta_*^{-2})$ with a schedule that advances $s$ at a constant rate, and
  $T=\mathcal{O}(\Delta_*^{-1})$ up to a polylogarithmic correction with a schedule that traverses the path at
  a constant geometric speed. The bounds on $L$
  hold for any twice continuously differentiable $H(s)$, and the bounds on $K$ for affine $H(s)$. We also prove the tightness of scaling of $L$ in $\Delta_*$ for every integer $p\ge1$, and
  the scaling of $K$ at $p=1$. We demonstrate the bounds on the adiabatic Grover search, the XXZ spin chain, and
  molecular electronic Hamiltonians.
\end{abstract}

\maketitle

\section{Introduction}
\label{sec:intro}
Adiabatic evolution is a cornerstone of quantum state preparation~\cite{Farhi2000,Aspuru-Guzik2005}. If a Hamiltonian
$H(s)$ keeps a nonzero energy gap above its ground state and is varied slowly enough,
the system started at the ground state of the initial Hamiltonian stays close to the instantaneous ground state
$\Phizeros$~\cite{BornFock1928,Kato1950,AvronSeilerYaffe1987,Nenciu1993,Teufel2003,Jansen2007,Albash2018}. These ground states trace
a curve in Hilbert space, the adiabatic path. The length of this curve is the adiabatic path length
$L=\int_0^1\|\partial_s\Phizeros\|\,ds$~\cite{ProvostVallee1980,Boixo2009,Rezakhani2010}, and its total curvature is
$K=\int_0^L\kappa\,dl$. Here the phase of $\Phizeros$ is chosen so that
$\langle\Phi_0|\partial_s\Phi_0\rangle=0$, and $l$ is the arc length, so $dl=\|\partial_s\Phizeros\|\,ds$. The pointwise
curvature is $\kappa=\|Q\,d^2|\Phi_0\rangle/dl^2\|$ with
$Q=\mathbb I-|\Phi_0\rangle\langle\Phi_0|$~\cite{Alsing2024,mchanCGS}.
The path geometry enters the analysis of $T$ in two ways. First, it sets a
limit. The evolution time $T$ is at least of order $L/\Dstar$, where $\Dstar$
is the minimum gap over $s\in[0,1]$~\cite{Boixo2010,Albash2018}. Second, it
gives a way to build the schedule. Apart from special cases such as the
adiabatic Grover search, the gap function is not known in advance. A schedule
that exploits it therefore cannot be designed beforehand. By measuring the path
geometry as the evolution proceeds, the constant geometric speed
schedule~\cite{mchanCGS} needs only a lower bound on the energy gap over
$s\in[0,1]$, whereas other schedules with similar performance
gain~\cite{Roland2002,Isermann2021,Braida2025,Guo2025} are built from the gap function along
the whole path. It lowers the upper bound on $T$ by one order in $\Dstar^{-1}$, when $L$ and $K$
stay bounded independently of $\Dstar$. A guarantee written through the path
geometry is therefore only as good as the bounds available on $L$ and $K$.

For the adiabatic path length $L$, the conventional bound is $L=\mathcal O(\Dstar^{-1/2})$~\cite{Chiang2014}.
The same conventional approach gives $K=\mathcal O(\Dstar^{-1})$ for the total
curvature $K$. These conventional bounds are, in practice, too pessimistic. For instance, in our earlier work~\cite{mchanCGS} the measured path length and total
curvature grew far more slowly with decreasing $\Dstar$ than these bounds allow. Other
work already noted that in many cases of interest $L$ can be bounded independently of
the gap~\cite{Boixo2009,Boixo2010}. The conventional bounds are loose for a simple reason. The bounds on $L$ and $K$ are both obtained from the gap integral
$\int_0^1 ds/\Delta(s)$.
The conventional bounds use only $\Delta(s)\ge\Dstar$, which gives
$\int_0^1 ds/\Delta(s)\le1/\Dstar$. The whole functional form of $\Delta(s)$ is thereby
reduced to the single number $\Dstar$. This makes the conventional bounds far above the actual $L$ and $K$ in many cases.
If the gap is small only near its minimum, the integral becomes far smaller than
$1/\Dstar$.

We close this discrepancy between the conventional bounds and the actual $L$ and $K$
by introducing the width class. It grades gap profiles by how fast
$\mubd(\gamma)$~\cite{Jarret2019,Guo2025} vanishes with $\gamma$, where $\mubd(\gamma)$ is the
width of the region in $s$ on which the gap $\Delta(s)$ is smaller than $\gamma$. Thus $\Dstar$ measures how deep the gap minimum is, and the width class how narrow that
region is. Previous works~\cite{Jarret2019,Guo2025} use $\mubd(\gamma)$ with the evolution time as the quantity of
interest, under the linear condition $\mubd(\gamma)\le C\gamma$. Here the quantities of interest are
the path length and the total curvature, and
the width class grades that condition, covering forms of $\mubd(\gamma)$
beyond the linear one.
If $\mubd(\gamma)$ vanishes at least linearly in
$\gamma$, the gap profile is in width class $p=1$, and the gap integral grows only
logarithmically as $\Dstar\to0$. This gives
$L=\mathcal O\big(\sqrt{\log\Dstar^{-1}}\big)$ and $K=\mathcal O(\log\Dstar^{-1})$.
A typical avoided crossing is in this class. If $\mubd(\gamma)$ vanishes more slowly, the gap
profile is in width class $p>1$, and the gap integral is $\mathcal O(\Dstar^{(1-p)/p})$,
so $L=\mathcal O(\Dstar^{(1-p)/(2p)})$ and
$K=\mathcal O(\Dstar^{(1-p)/p})$.
In any case, including width class information improves the scaling of the bound
compared to that of the conventional bound, and the conventional exponents are recovered
in the limit $p\to\infty$. We also show that our scaling for each width class is
tight for $L$ with every integer $p\ge1$ and for $K$ with $p=1$, by constructing
Hamiltonian families that attain these scalings.

$L$ and $K$ often stay bounded by a constant. For example,
Sec.~\ref{sec:bounded} proves boundedness of $L$ and $K$ for affine Hamiltonians
whose path stays in a fixed two-dimensional subspace. The width class also lowers the bound on the
evolution time. For the linear schedule it gives
$T=\mathcal O(\Dstar^{-(3-1/p)})$, which is $\mathcal O(\Dstar^{-2})$ at $p=1$ in
place of the worst case $\mathcal O(\Dstar^{-3})$. For the constant geometric
speed schedule it gives $T=\mathcal O(\Dstar^{-(2-1/p)})$, which is one order
below the linear schedule in every width class and $\mathcal O(\Dstar^{-1})$ at
$p=1$, whenever $L$ stays bounded independently of $\Dstar$. Moreover, a gain over the
linear schedule is proved for the same schedule even when $L$ is unbounded
(Sec.~\ref{sec:runtime}).

We prove the bounds for
the ground subspace, the spectral subspace of the $g$ lowest levels, rather than for the
ground state alone. The quantities bounded are then the path length $L_P$ and the total
curvature $K_P$ of that subspace, which reduce to $L$ and $K$ at $g=1$. The gap is
measured above the subspace, as in the adiabatic theorem of Jansen, Ruskai, and
Seiler~\cite{Jansen2007}. The bounds on $L_P$ and $K_P$ hold even when the $g$ levels
inside the subspace cross.
We demonstrate the length and curvature bounds on the adiabatic Grover search,
the XXZ spin chain, and molecular electronic Hamiltonians
(Sec.~\ref{sec:applications}).

\section{Theory}
\label{sec:Theory}
We consider a Hamiltonian
\begin{equation}
  H(s)=H_0+sH',\qquad s\in[0,1],\label{eq:Hs}
\end{equation}
in a finite-dimensional Hilbert space $\mathcal H$. The dependence on $s$ is affine,
so the second derivative of $H(s)$ vanishes. Its eigenvalues are labeled from below and counted with multiplicity, $E_0(s)\le E_1(s)\le\cdots$, with $\Phijs$
a corresponding orthonormal eigenbasis chosen pointwise in $s$.
The eigenstates of the lowest $g$ eigenvalues form the ground subspace
[Fig.~\ref{fig:mechanism}(a)],
$P(s)$ denotes its spectral projector, and $Q(s)=\mathbb I-P(s)$ represents the spectral projector to its complement.
The gap between the ground subspace and the rest of the spectrum is
\begin{align}
  \Delta(s) = E_g(s) - E_{g-1}(s),
\end{align}
with a minimum $\Delta_*=\min_{s\in[0,1]}\Delta(s)>0$, attained at $s_*$. Let $\Gamma$ be an upper bound on $\Delta$, so $\Delta_*\le\Delta(s)\le\Gamma$ for all $s$.
Each eigenvalue changes at a rate at most $\Lambda=\|H'\|$. Indeed,
$\|H(s)-H(t)\|=\Lambda|s-t|$, so Weyl's perturbation theorem gives
$|E_j(s)-E_j(t)|\le\Lambda|s-t|$ for every $j$~\cite{Bhatia1997}.
The $g$ eigenvalues inside the ground subspace may split, touch, and cross. The
eigenvalues and eigenvectors labeled in increasing order can then be non-differentiable
in $s$. Only
$P(s)$, which inherits the smoothness of $H(s)$ through the contour-integral
representation of the spectral projector~\cite{Kato1995,Jansen2007}, is
differentiated in what follows.
We define the speed of the ground subspace as
\begin{align}
  v_P(s) = \frac{1}{\sqrt{2}} \|P'(s)\|_{\HS},
  \label{eq:vpdef}
\end{align}
where $'$ represents the derivative with respect to $s$ and
$\|A\|_{\HS}=\sqrt{\operatorname{Tr}(A^\dagger A)}$ is the norm derived from the Hilbert--Schmidt inner product.
The infinitesimal motion of the ground subspace decomposes into the rotation rates of the
principal angles~\cite{BjorckGolub1973,EdelmanAriasSmith1998}, and $v_P(s)$ is the
Euclidean norm of those rates.
When $g=1$, $v_P$ reduces to the speed of traversing the adiabatic path of
$\Phizeros$.
The $s$ dependence of variables is omitted where convenient.

\subsection{Bound for the path length}
\label{sec:path-length-bound}

The path length is defined as
\begin{align}
  L_P = \int_0^1 v_P(s) \, ds.
  \label{eq:LPdef}
\end{align}
Thus, $L_P$ is the angular distance accumulated by the ground subspace.
Then, with
\begin{align}
  M_1(s) = \sum_{j<g\le k}\frac{|\langle\Phi_k|H'|\Phi_j\rangle|^2}{E_k-E_j},\label{eq:M1s}
\end{align}
we have the following lemma for $v_P$.

\begin{lem}\label{lem:speed}
With $M_1(s)$ defined as in Eq.~\eqref{eq:M1s}, 
\begin{align}
v^2_P(s)\le \frac{M_1(s)}{\Delta(s)}.
\end{align}
\end{lem}
\begin{proof}
  From $P^2=P$, we have $PP'+P'P=P'$. Multiplying by $P$ from the right gives
  $PP'P = 0$. Likewise, from $QP=0$, we have $QP'Q=0$.
  Then, $P+Q=1$ gives
  \begin{align*}
    P' = PP'P + PP'Q + QP'P + QP'Q = X+X^\dagger,
  \end{align*}
  where $X = QP'P$.
  Therefore, 
  \begin{align}
    v_P^2 = \frac{1}{2} \operatorname{Tr}((X^\dagger+X)^2) = \operatorname{Tr}(X^\dagger X),
    \label{eq:vp2x}
  \end{align}
  where $X^2=0$ because $X$ maps $\Ran P$ into $\Ran Q$.
  What remains is to find $X$.
  Differentiating $[H,P]=0$ gives
  \begin{align*}
    H'P + HP' - P'H - PH' = 0
  \end{align*}
  Multiplying by $Q$ from the left and by $P$ from the right, we have
  \begin{align*}
    QH'P + QHP'P-QP'HP = 0, \\
    QHQX - XPHP = -QH'P.
  \end{align*}
  Because $QHQ$ and $PHP$ have disjoint spectra, this equation is solvable for $X$ and gives
  \begin{align}
    X_{kj} = -\frac{\langle \Phi_k |H'|\Phi_j\rangle}{E_k-E_j}, \quad \text{for } j<g\le k, \label{eq:Xsol}
  \end{align}
  with all other elements of zeros.
  Combining Eqs.~\eqref{eq:vp2x} and \eqref{eq:Xsol} gives
  \begin{align}
    v_P^2 = \sum_{j<g\le k} \frac{|\langle \Phi_k |H'|\Phi_j\rangle|^2}{(E_k-E_j)^2}.
    \label{eq:vPexact}
  \end{align}
  $E_k-E_j \geq \Delta(s)$ proves the lemma.
\end{proof}

We denote the integral of $M_1$ by
\begin{align}
  C_L=\int_0^1 M_1(s)\,ds.
  \label{eq:CLdef}
\end{align}
The following lemma
evaluates it and bounds it by the eigenvalues of $H'$. To state
this, we define $\Lambda_g$ as half the difference between the sums of the $g$
largest and $g$ smallest eigenvalues of $H'$, so that $\Lambda_g\le g\Lambda$.
\begin{lem}\label{lem:budget}
Let $\mathcal{E}_P$ be the sum of the $g$ lowest energy eigenvalues of $H$, so $\mathcal{E}_P=\operatorname{Tr}(PH)$. Then,
\begin{align}
\mathcal{E}_P''(s)&=-2M_1(s), \nonumber\\
C_L &=\tfrac12\big[\mathcal{E}_P'(0)-\mathcal{E}_P'(1)\big]
\le \Lambda_g.
\end{align}
\end{lem}
\begin{proof}
  From $P'=PP'Q+QP'P$, as in the proof of Lemma~\ref{lem:speed}, and $[H,P]=0$,
  \begin{align*}
    \mathcal{E}_P' = \operatorname{Tr}(PH' + P'H) = \operatorname{Tr}(PH').
  \end{align*}
  So, using \eqref{eq:Xsol},
  \begin{align*}
    \mathcal{E}_P'' &= \operatorname{Tr}(P'H') = \operatorname{Tr}((X+X^\dagger)H') \\
                    &= -2 \sum_{j<g\le k}\frac{|\langle\Phi_k|H'|\Phi_j\rangle|^2}{E_k-E_j} = -2 M_1(s).
  \end{align*}
  Then,
  \begin{align*}
    C_L &= -\frac{1}{2} \int_0^1 \mathcal{E}_P'' ds = \frac{1}{2} (\mathcal{E}_P'(0) - \mathcal{E}_P'(1)) \\
                       &= \frac{1}{2} [\operatorname{Tr}(P(0)H')-\operatorname{Tr}(P(1)H')] \leq \Lambda_g,
  \end{align*}
  because $\operatorname{Tr}(P(0)H')$ is at most the sum of the $g$ largest
  eigenvalues of $H'$ and $\operatorname{Tr}(P(1)H')$ is at least the sum of the $g$
  smallest~\cite{Bhatia1997}.
\end{proof}
Lemmas~\ref{lem:speed} and \ref{lem:budget} combine to bound $L_P$ by the integral of $1/\Delta(s)$.
\begin{thm}
\begin{align}
  L_P&\le\Big(C_L\int_0^1\frac{ds}{\Delta(s)}\Big)^{1/2}\nonumber\\
     &\le\Big(\Lambda_g\int_0^1\frac{ds}{\Delta(s)}\Big)^{1/2}.
\end{align}
\label{thm:L}
\end{thm}
\begin{proof}
  Inserting Lemma~\ref{lem:speed} into Eq.~\eqref{eq:LPdef} and using the Cauchy--Schwarz inequality, we have
  \begin{align*}
    L_P \leq \int_0^1 \left(\frac{M_1(s)}{\Delta(s)}\right)^{1/2} ds \leq \left(C_L \int_0^1 \frac{ds}{\Delta(s)}\right)^{1/2}.
  \end{align*}
  Then, $C_L\le\Lambda_g$ from Lemma~\ref{lem:budget} proves the second inequality.
\end{proof}
Using $\Delta(s)\ge\Delta_*$, we have
\begin{align}
  L_P \leq \left(\frac{C_L}{\Delta_*}\right)^{1/2}, \label{eq:crudepathlengthbound}
\end{align}
which recovers, at any rank $g$, the best known scaling $\mathcal{O}(\Delta_*^{-1/2})$ of the adiabatic path length $L$~\cite{Chiang2014}.
This bound replaces the integral of $1/\Delta(s)$ over $s\in[0,1]$ by its maximum value.
The actual value of the integral can be much smaller than $1/\Dstar$, depending on how the small-gap region behaves.

In the numerical experiments underlying our previous work~\cite{mchanCGS}, we observed that the width
of this region shrinks as the threshold defining a small gap is lowered.
How fast this region shrinks is what the gap integral is sensitive to. For instance,
a width vanishing linearly in the threshold renders the integral logarithmic,
while one vanishing more slowly gives a power law. We therefore classify gap profiles by how fast
the width shrinks.
\begin{defn}\label{def:width}
The gap profile is in \emph{width class} $p\ge1$ if
\begin{align}
\mubd(\gamma):=\left|\{s:\Delta(s)<\gamma\}\right|
\;\le\; C_W\left(\frac{\gamma}{\Gamma}\right)^{1/p}
\end{align}
with a positive $C_W$ for all $\gamma$ in $[\Delta_*,\Gamma]$.
\end{defn}

Definition~\ref{def:width} applies to a single Hamiltonian, so
$\Delta_*$ is one fixed number, while we want to describe how $L_P$ and
$K_P$ behave as $\Delta_*$ decreases. We therefore consider a set of
Hamiltonians with varying $\Delta_*$. In
our earlier work~\cite{mchanCGS}, for example, one such set consists of molecular
Hamiltonians at different bond lengths, with $\Delta_*$ spanning several
decades. We call such a set a family.
\begin{defn}\label{def:family}
A \emph{family} is a set of twice continuously differentiable Hamiltonians
$H(s)$ satisfying the assumptions of this section except
Eq.~\eqref{eq:Hs}, with a common rank $g$ of the ground subspace. A family is
\emph{affine} if every member is of the form~\eqref{eq:Hs}. An affine family is in
\emph{width class} $p$ if there exist constants $\Gamma$ and $C_W$ such that every
member satisfies Definition~\ref{def:width} with this same pair $(\Gamma,C_W)$, and
if a single constant bounds $\Lambda=\|H'\|$ for every member.
\end{defn}

Definition~\ref{def:width} is not the only constraint on $\mubd$. Because
$\{s:\Delta(s)<\gamma\}$ is a subset of $[0,1]$, $\mubd(\gamma)\le1$.
The two constraints cross at
$\Gamma C_W^{-p}$. Because $\gamma\le\Gamma$, together they give
$\mubd(\gamma)\le C_W(\gamma/\Gamma)^{1/p}$ for $\gamma\le\gamma_c$ and
$\mubd(\gamma)\le1$ for $\gamma>\gamma_c$, with
$\gamma_c=\min\{\Gamma C_W^{-p},\Gamma\}$.
The following lemma bounds the gap integral of Theorem~\ref{thm:L}.
\begin{lem}\label{lem:layer}
(i) For any gap profile,
\begin{align}
\int_0^1\frac{ds}{\Delta(s)}
=\frac{1}{\Gamma}+\int_{\Dstar}^{\Gamma}\frac{\mubd(\gamma)}{\gamma^2}\,d\gamma.
\end{align}
(ii) If the gap profile is in width class $p$, then
$\int_0^1 ds/\Delta(s)\le I_p(\Dstar)$, where $I_p(\Dstar)=1/\Dstar$ for
$\Dstar>\gamma_c$, and for $\Dstar\le\gamma_c$,
\begin{align}
I_p(\Dstar)&=\frac{1}{\gamma_c}+\frac{C_W}{\Gamma}\log\frac{\gamma_c}{\Dstar}
&& (p=1)
\nonumber\\
&=\frac{1}{\gamma_c}+\frac{pC_W\big(\Dstar^{\frac{1-p}{p}}-\gamma_c^{\frac{1-p}{p}}\big)}{(p-1)\,\Gamma^{1/p}}
&& (p>1).
\label{eq:layerclipped}
\end{align}
In all cases $I_p(\Dstar)\le\Dstar^{-1}$.
\end{lem}
\begin{proof}
For every $s$,
\begin{align*}
  \frac{1}{\Delta(s)}
  =\int_{\Delta(s)}^\infty\frac{d\gamma}{\gamma^2}
  =\int_{\Delta_*}^\infty
    \frac{\mathds{1}_{\gamma>\Delta(s)}}{\gamma^2}\,d\gamma .
\end{align*}
Integrating this over $s$ and exchanging the order of the two integrals,
\begin{align*}
  \int_0^1 \frac{ds}{\Delta(s)}
  &=\int_{\Delta_*}^\infty\frac{d\gamma}{\gamma^2}\int_0^1\mathds{1}_{\gamma>\Delta(s)}\,ds
   =\int_{\Delta_*}^\infty\frac{\mubd(\gamma)}{\gamma^2}\,d\gamma \\
  &=\int_{\Delta_*}^{\Gamma}\frac{\mubd(\gamma)}{\gamma^2}\,d\gamma
    +\int_{\Gamma}^\infty\frac{\mubd(\gamma)}{\gamma^2}\,d\gamma \\
  &= \int_{\Delta_*}^{\Gamma}\frac{\mubd(\gamma)}{\gamma^2}\,d\gamma
    +\frac{1}{\Gamma},
\end{align*}
where the inner $s$-integral is $\int_0^1\mathds{1}_{\gamma>\Delta(s)}\,ds=\mubd(\gamma)$
and the last step uses $\mubd(\gamma)=1$ for $\gamma>\Gamma$, which holds
since $\Delta(s)\le\Gamma$ for all $s$.

For (ii), let $\Dstar\le\gamma_c$ and split the integral of (i) at $\gamma_c$.
On $[\gamma_c,\Gamma]$, $\mubd\le1$ gives
$\int_{\gamma_c}^{\Gamma}\gamma^{-2}d\gamma=1/\gamma_c-1/\Gamma$. On
$[\Dstar,\gamma_c]$, Definition~\ref{def:width} gives
\begin{align*}
&\int_{\Dstar}^{\gamma_c}\frac{\mubd(\gamma)}{\gamma^2}\,d\gamma
\le \frac{C_W}{\Gamma^{1/p}}\int_{\Dstar}^{\gamma_c}\gamma^{\frac1p-2}\,d\gamma\\
&\quad=
\begin{cases}
\dfrac{C_W}{\Gamma}\log\dfrac{\gamma_c}{\Dstar}, & p=1,\\[8pt]
\dfrac{p\,C_W\big(\Dstar^{-\frac{p-1}{p}}-\gamma_c^{-\frac{p-1}{p}}\big)}{(p-1)\,\Gamma^{1/p}}, & p>1.
\end{cases}
\end{align*}
Adding the three parts gives $I_p(\Dstar)$ for $\Dstar\le\gamma_c$. For
$\Dstar>\gamma_c$, $\mubd\le1$ on all of $[\Dstar,\Gamma]$, and (i) gives
$1/\Dstar=I_p(\Dstar)$, and the two branches of $I_p$ meet at $\Dstar=\gamma_c$.
Finally, because $C_W(\gamma/\Gamma)^{1/p}\le1$ for $\gamma\le\gamma_c$, the same
integrals give $I_p(\Dstar)\le\Dstar^{-1}$ in all cases.
\end{proof}
Combining Theorem~\ref{thm:L} and Lemma~\ref{lem:layer}, we have
\begin{thm}\label{thm:rates}
  Suppose the gap profile of the ground subspace of $H(s)$ is in width class $p$.
  Then $L_P\le[C_L\,I_p(\Dstar)]^{1/2}$.
  Therefore, along an affine family of width class $p$,
  $L_P = \mathcal{O}\big(\sqrt{\log \Delta_*^{-1}}\big)$ for $p=1$
  and $L_P=\mathcal{O}(\Delta_*^{\frac{1-p}{2p}})$ for $p>1$.
\end{thm}
\begin{proof}
Insert Lemma~\ref{lem:layer}(ii) into Theorem~\ref{thm:L}.
\end{proof}
Theorem~\ref{thm:rates} improves on \cref{eq:crudepathlengthbound} in every width
class, because $I_p(\Dstar)\le1/\Dstar$. At $p=1$ Theorem~\ref{thm:rates} replaces the power law by
$\sqrt{\log\Delta_*^{-1}}$, and at $p>1$ it lowers the exponent of $1/\Dstar$
from $1/2$ to $(p-1)/(2p)$.

\begin{figure}
  \centering
  \includegraphics[width=\columnwidth]{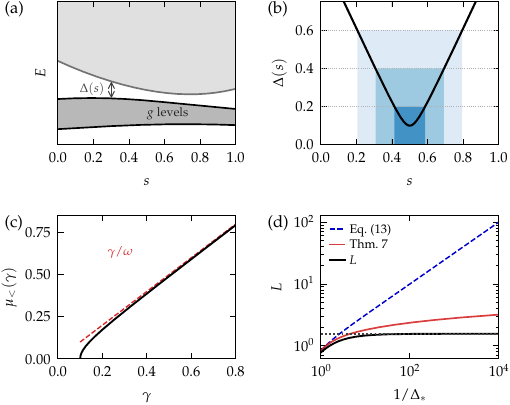}
  \caption{(a) A schematic of the ground subspace. The $g$ lowest levels are in
  the darker shaded region, inside which they may cross or become degenerate, and
  $\Delta(s)$ separates them from the rest of the spectrum. (b)--(d) The
  two-level crossing
  $H(s)=\omega(s-s_*)\sigma_z+\tfrac{\Delta_*}{2}\sigma_x$ with $\omega=1$ and
  $s_*=1/2$, at $\Delta_*=0.1$ in (b) and (c). (b) The shaded regions
  represent the sets $\{s:\Delta(s)<\gamma\}$ for $\gamma=0.2,0.4,0.6$. (c) Their width
  $\mubd(\gamma)$ (solid) and the bound $\gamma/\omega$ of
  \cref{eq:minorant_mubd} (dashed). (d) $L$ versus $1/\Delta_*$, with
  \cref{eq:crudepathlengthbound} (blue), the bound of
  Theorem~\ref{thm:rates} (red), and the computed $L$ (black). The dotted line
  is $\pi/2$.}
  \label{fig:mechanism}
\end{figure}

The inequality in Definition~\ref{def:width} reads the same way for $0<p<1$, and
one might ask whether an affine family of that class could give $L_P=\mathcal O(1)$.
For $0<p<1$, the integrand $\mubd(\gamma)/\gamma^2$ in Lemma~\ref{lem:layer}(i) is at
most $C_W\Gamma^{-1/p}\gamma^{1/p-2}$, whose integral over $[0,\Gamma]$ is finite. The
gap integral then stays bounded as $\Dstar\to0$, and Theorem~\ref{thm:L} gives
$L_P=\mathcal O(1)$. However, an affine family with $\Dstar\to0$ cannot be in width class $0<p<1$.
Each eigenvalue moves at a rate of at most $\Lambda$, so $\Delta(s)\le\Dstar+2\Lambda|s-s_*|$, which is
below $2\Dstar$ for $|s-s_*|<\Dstar/(2\Lambda)$. Within $[0,1]$, these $s$ form an interval of
length $\min\{s_*,\Dstar/(2\Lambda)\}+\min\{1-s_*,\Dstar/(2\Lambda)\}$.
Because one of $s_*$ and $1-s_*$ is at least $\tfrac12$, we have
\begin{align}
  \mubd(2\Dstar)\ge\min\Big\{\frac12,\frac{\Dstar}{2\Lambda}\Big\}.
  \label{eq:mubdlower}
\end{align}
It also holds when $\Lambda=0$, where the gap is constant and $\mubd(2\Dstar)=1$.
A single constant bounds $\Lambda$ along the family (Definition~\ref{def:family}), so the
right-hand side of \cref{eq:mubdlower} is at least a constant times $\Dstar$ as
$\Dstar\to0$. Definition~\ref{def:width} with $0<p<1$ would instead give
$\mubd(2\Dstar)\le C_W(2\Dstar/\Gamma)^{1/p}=o(\Dstar)$, which contradicts \cref{eq:mubdlower}. Thus $p=1$ is extremal, though the actual $L_P$
of a $p=1$ family may still stay bounded, as in Sec.~\ref{sec:applications}.

Applying Theorem~\ref{thm:rates} requires the width class information, which can
be complicated to check. The following lemma derives it from a simple condition
on the gap.
\begin{lem}\label{lem:minorant}
Denote by $s_*$ a minimizer of $\Delta$, so $\Delta(s_*)=\Delta_*$.
Suppose that, for a constant $\omega>0$,
\begin{align}
  \Delta(s)\ge2\omega\,|s-s_*|^{p},
  \qquad s\in[0,1].
  \label{eq:minorant}
\end{align}
Then, for $\gamma\in[\Dstar,\Gamma]$,
\begin{align}
  \mubd(\gamma)\;\le\;
  2\left(\frac{\gamma}{2\omega}\right)^{1/p},\label{eq:minorant_mubd}
\end{align}
so the gap profile is in width class $p$ with
$C_W=2(\Gamma/2\omega)^{1/p}$.
\end{lem}
\begin{proof}
  If $\Delta(s)<\gamma$, Eq.~\eqref{eq:minorant} implies
  $2\omega|s-s_*|^p<\gamma$, so $\{s:\Delta(s)<\gamma\}$ lies in an interval of
  length at most $2(\gamma/2\omega)^{1/p}$, which is Eq.~\eqref{eq:minorant_mubd}.
  Writing $2(\gamma/2\omega)^{1/p}=C_W(\gamma/\Gamma)^{1/p}$ with
  $C_W=2(\Gamma/2\omega)^{1/p}$ gives Definition~\ref{def:width}.
\end{proof}

From this lemma, we see that width class $p=1$ deserves particular attention.
The reason is that it corresponds to a so-called avoided crossing.
Consider two levels that a symmetry leaves uncoupled, so that they can cross linearly in
$s$. Once that symmetry is broken, the coupling separates them and the crossing becomes
avoided. The two-level form of the gap near it is
$\sqrt{\Dstar^2+4\omega^2(s-s_*)^2}$. This is at least $2\omega|s-s_*|$, so
Eq.~\eqref{eq:minorant} holds with $p=1$ and
$L_P=\mathcal O\big(\sqrt{\log\Dstar^{-1}}\big)$.
Figure~\ref{fig:mechanism}(b,c) shows the sets $\{s:\Delta(s)<\gamma\}$ and the width
$\mubd(\gamma)$ for such a crossing, together with the bound
\cref{eq:minorant_mubd}. Avoided crossings are common in
physical and chemical systems, and every family of Sec.~\ref{sec:applications} is in
width class $p=1$.

The scaling of Theorem~\ref{thm:rates} in $\Dstar$ is tight for every integer
$p\ge1$. The tightness is proved by the existence of an
affine family of width class $p=1$ whose path length grows as
$\sqrt{\log\Dstar^{-1}}$, and, for each integer $p\ge2$, of one whose path
length grows as $\Dstar^{-(p-1)/(2p)}$ (Appendix~\ref{app:sharpness}).

\subsection{Bound for the total curvature}
\label{sec:curvature}

The adiabatic runtime bound for the constant geometric speed
schedule~\cite{mchanCGS} depends on the product $L(K+L)$. Here $K=\int_0^L\kappa\,dl$ is the total curvature of the
path, $l$ is the arc length, and $dl=v_P\,ds$. The pointwise curvature is
\begin{align}
  \kappa=\left\|Q\,\frac{d^2|\Phi_0\rangle}{dl^2}\right\|,
  \label{eq:kappa}
\end{align}
which vanishes along a geodesic of the state space. We now bound the corresponding quantity for the ground subspace.

The curvature involves second derivatives of the states. To deal with this, we
work with a moving frame $U(s)$, a matrix
whose $g$ orthonormal columns span the ground subspace. It satisfies
$U^\dagger U=\mathbb I_g$ and $UU^\dagger=P$. We choose the parallel-transport
gauge $U^\dagger U'=0$, which removes rotations of the frame within the ground
subspace. Such a frame is obtained
by transporting an initial frame, $U(s)=W(s)U(0)$, where the unitary $W$
solves $W'=[P',P]\,W$ with $W(0)=\mathbb I$~\cite{Kato1950,AvronSeilerYaffe1987}. Then $P(s)W(s)=W(s)P(0)$, so the
columns of $U(s)$ stay in the ground subspace, and $U'=[P',P]U=P'U$ because
$PU=U$ and $PP'P=0$. Writing $P'=X+X^\dagger$ with $X=QP'P$ as in the proof of
Lemma~\ref{lem:speed}, this reads $U'=XU$, which maps into $\Ran Q$. The gauge
condition therefore holds, and $\|U'\|_{\HS}=\|X\|_{\HS}=v_P$.

The unit tangent of the motion is $u=U'/v_P$. This requires $v_P>0$,
which the following lemma guarantees whenever the projector $P(s)$ is not
constant on $[0,1]$.
\begin{lem}\label{lem:neverstops}
If $v_P(s_0)=0$ at some $s_0\in[0,1]$, then $P(s)$ is constant on $[0,1]$.
\end{lem}
\begin{proof}
At $s_0$, Eq.~\eqref{eq:vPexact} gives
$\langle\Phi_k|H'|\Phi_j\rangle=0$ for all $j<g\le k$, that is, $QH'P=0$.
Taking the adjoint, $PH'Q=0$. Thus $P(s_0)$ commutes with $H'$. It also commutes
with $H(s_0)$, being a spectral projector of $H(s_0)$. Since
$H(s)=H(s_0)+(s-s_0)H'$ by Eq.~\eqref{eq:Hs}, $P(s_0)$ commutes with $H(s)$ for
every $s$. The spectrum then splits into two blocks, one carried by $\Ran P(s_0)$ and one
by its complement. Let $e_P(s)$ be the largest eigenvalue of the first
block and $e_Q(s)$ the smallest of the second, both continuous in $s$.
Wherever $e_Q(s)>e_P(s)$, the first block holds the lowest $g$
eigenvalues. Then $e_P=E_{g-1}$, $e_Q=E_g$, and
$e_Q-e_P=\Delta(s)\ge\Dstar$. This is the case at $s_0$. If
$e_Q(s)\le e_P(s)$ somewhere, let $s_1$ be such a point closest to
$s_0$. Then $e_Q-e_P\ge\Dstar$ between $s_0$ and $s_1$, and
continuity gives $e_Q(s_1)-e_P(s_1)\ge\Dstar$, a contradiction.
Hence the first block holds the lowest $g$ eigenvalues for every $s$, and
$P(s)=P(s_0)$.
\end{proof}

In analogy with Eq.~\eqref{eq:kappa}, the total curvature of the ground
subspace is
\begin{align}
  K_P=\int_0^1\big\|Q\,u'\big\|_{\HS}\,ds.
  \label{eq:KPdef}
\end{align}
Any two frames in our gauge choice differ by a constant unitary, so $K_P$ does
not depend on the frame choice. At $g=1$ the frame is the ground state itself,
with its phase fixed by $\langle\Phi_0|\Phi_0'\rangle=0$. Then
$u=d|\Phi_0\rangle/dl$, $\|Qu'\|\,ds=\kappa\,dl$, and $K_P=K$. For a constant
$P$ the path does not turn, and we set $K_P=0$.

The bound for $K_P$ comes from the relation $HU=Uh$, where $h=U^\dagger HU$ is
the Hamiltonian in the frame. It is stated in terms of
\begin{align}
  Y=QH'u-uh',
  \label{eq:Ydef}
\end{align}
and the average of $\|Y\|_{\HS}$ with weight $1/\Delta(s)$,
\begin{align}
  C_K=\frac{\displaystyle\int_0^1\|Y\|_{\HS}\,\frac{ds}{\Delta(s)}}
           {\displaystyle\int_0^1\frac{ds}{\Delta(s)}},
  \label{eq:CKdef}
\end{align}
which is defined whenever $P$ is not constant, since $u$ is then defined on all of
$[0,1]$ by Lemma~\ref{lem:neverstops}. For a constant $P$ we set $C_K=0$.
\begin{thm}\label{thm:K}
Let $\Omega$ be the difference between the largest and the smallest
eigenvalues of $H'$. Then,
\begin{align}
  K_P&\le2\,C_K\int_0^1\frac{ds}{\Delta(s)}\nonumber\\
     &\le2\,\Omega\int_0^1\frac{ds}{\Delta(s)}.
  \label{eq:KP}
\end{align}
\end{thm}
\begin{proof}
For a constant $P$, $K_P=0$. Otherwise $v_P>0$ by
Lemma~\ref{lem:neverstops}, and $u$ is defined on all of $[0,1]$.
Let $h=U^\dagger HU$, so $HU=Uh$.
Differentiating this identity and multiplying by $U^\dagger$ from the left gives
$h'=U^\dagger H'U$, using $U^\dagger U'=0$ and $U^\dagger H=hU^\dagger$.
Differentiating $HU=Uh$ twice gives
\begin{align*}
  H''U+2H'U'+HU''=U''h+2U'h'+Uh'' ,
\end{align*}
whose first term vanishes because $H(s)$ is affine. Multiply what is left by
$Q$ from the left. The term $Uh''$ drops because $QU=0$, while
$QHU''=(QHQ)(QU'')$ because $QHP=0$, and $QU'h'=U'h'$ because $QU'=U'$.
With the map
$\mathcal S(X):=(QHQ)X-Xh$ on matrices from $\mathbb C^g$ to $\Ran Q$, the
remaining identity becomes
\begin{align}
  \mathcal S(QU'')=-2\,\big(QH'U'-U'h'\big).
  \label{eq:sylvester}
\end{align}
The map $\mathcal S$ is
self-adjoint for the Hilbert--Schmidt inner product
$\langle A,B\rangle_{\HS}=\operatorname{Tr}(A^\dagger B)$. Write $\chi_j$ for the
eigenvectors of $h$, with eigenvalues $E_j$, $j<g$. The matrices
$|\Phi_k\rangle\langle\chi_j|$, $k\ge g$, are eigenmatrices of
$\mathcal S$. There are $(\dim\mathcal H-g)\,g$ of them, a basis of the space
$\mathcal S$ acts on.
The eigenvalues are $E_k-E_j\ge\Delta$, so $\mathcal S$ is invertible with
\begin{align}
  \|\mathcal S^{-1}(X)\|_{\HS}\le\frac{\|X\|_{\HS}}{\Delta(s)}
  \label{eq:Sinv}
\end{align}
for every $X$.
Insert $U'=v_Pu$ and $QU''=v_P'\,u+v_P\,Qu'$ into Eq.~\eqref{eq:sylvester},
using $Qu=u$. Its right-hand side becomes $-2v_PY$ with $Y$ as in
Eq.~\eqref{eq:Ydef}. Applying $\mathcal S^{-1}$ and dividing by $v_P$, with
$Z=\mathcal S^{-1}(Y)$,
\begin{align}
  \frac{v_P'}{v_P}\,u+Qu'=-2Z.
  \label{eq:turningraw}
\end{align}
Because $\|u\|_{\HS}=1$, $\operatorname{Re}\operatorname{Tr}(u^\dagger u')=0$.
Also, $u^\dagger Qu'=u^\dagger u'$ because $Qu=u$.
The real part of the Hilbert--Schmidt inner product of
Eq.~\eqref{eq:turningraw} with $u$ therefore gives
$v_P'/v_P=-2\operatorname{Re}\operatorname{Tr}(u^\dagger Z)$. Substituting it
back into Eq.~\eqref{eq:turningraw},
\begin{align}
  Qu'=-2\,\big[Z-u\,\operatorname{Re}\operatorname{Tr}(u^\dagger Z)\big].
  \label{eq:turning}
\end{align}
The subtracted term removes the component of $Z$ along $u$, so
$\|Qu'\|_{\HS}\le2\|Z\|_{\HS}$. With Eq.~\eqref{eq:Sinv},
\begin{align}
  \|Qu'\|_{\HS}\le\frac{2\,\|Y\|_{\HS}}{\Delta(s)},
  \label{eq:turningpointwise}
\end{align}
and integrating this over $s$ is the first inequality, by the definition
of $C_K$ in Eq.~\eqref{eq:CKdef}.

For the second, let $c$ be the
midpoint of the spectrum of $H'$, so $\|H'-c\|=\Omega/2$. Since
$h'=U^\dagger H'U$, its spectrum lies in the same interval, so
$\|h'-c\mathbb I_g\|\le\Omega/2$. Then
\begin{align*}
  Y=Q(H'-c)u-u\,(h'-c\mathbb I_g),
\end{align*}
and each term has Hilbert--Schmidt norm at most $\Omega/2$, because
$\|AX\|_{\HS}\le\|A\|\,\|X\|_{\HS}$ and
$\|XB\|_{\HS}\le\|X\|_{\HS}\,\|B\|$, with $\|u\|_{\HS}=1$ and
$\|Q\|\le1$. Hence $\|Y\|_{\HS}\le\Omega$ pointwise, and $C_K$, an
average of $\|Y\|_{\HS}$, obeys the same bound.
\end{proof}
The counterpart of \cref{eq:crudepathlengthbound} for the total curvature is obtained the
same way, by replacing the gap profile with its minimum,
\begin{align}
  K_P\le\frac{2C_K}{\Dstar}.
  \label{eq:crudecurvaturebound}
\end{align}
We refer to \cref{eq:crudepathlengthbound} and \cref{eq:crudecurvaturebound} together as
the conventional bounds.

The gap enters Eq.~\eqref{eq:KP} through the same integral as in
Theorem~\ref{thm:L}, so Lemma~\ref{lem:layer}(ii)
converts it into a scaling in $\Dstar$, now without the square root.
\begin{cor}\label{cor:Krates}
Suppose the gap profile is in width class $p$. Then
$K_P\le2\,C_K\,I_p(\Dstar)$.
Therefore, along an affine family of width class $p$, $K_P=\mathcal O(\log\Dstar^{-1})$ for $p=1$ and
$K_P=\mathcal O(\Dstar^{\frac{1-p}{p}})$ for $p>1$.
\end{cor}
At $p=1$ the bounds scale as $\sqrt{\log\Dstar^{-1}}$ for the length and as
$\log\Dstar^{-1}$ for the curvature, and this pair is tight. A single family of
width class $p=1$ saturates both at once
(Theorems~\ref{thm:stair-length} and~\ref{thm:stair-curvature}).

\subsection{Bound for non-affine Hamiltonians}

Until now we derived the bounds for affine $H(s)$. In this subsection we
consider non-affine $H(s)$. Definition~\ref{def:width} and the notion of a
family in Definition~\ref{def:family} apply without change.

Affinity entered the length bound only through Lemma~\ref{lem:budget}, and
there only through the bound on $C_L$, not through its definition. It is
therefore enough to replace that lemma. Its
proof rests on the identity $\mathcal E_P''=-2M_1$, and for non-affine
$H(s)$ the second derivative of the energy acquires one more term, so that
\begin{align}
  \mathcal{E}_P''=\operatorname{Tr}(PH'')-2M_1.
  \label{eq:EPcurved}
\end{align}
Because $P$ is a rank-$g$ orthogonal projector and $A$ is Hermitian,
$\operatorname{Tr}(PA)$ lies between the sums of the $g$ smallest and the $g$
largest eigenvalues of $A$~\cite{Bhatia1997}, which we denote $\ell_g(A)$
and $r_g(A)$. Therefore
$|\operatorname{Tr}(PA)|\le\beta_g(A):=\max\{|\ell_g(A)|,|r_g(A)|\}$,
and $\beta_1(A)=\|A\|$. Accordingly, we define
$\mathcal{B}_g=\int_0^1\beta_g(H''(s))\,ds$, and generalize $\Lambda_g$
of Lemma~\ref{lem:budget} to
\begin{align}
  \Lambda_g=\tfrac12\big[r_g(H'(0))-\ell_g(H'(1))\big].
  \label{eq:nonaffinebudgets}
\end{align}
For affine $H(s)$, $\mathcal{B}_g=0$ and the two definitions of $\Lambda_g$
agree.

\begin{lem}\label{lem:budgetcurved}
For twice continuously differentiable $H(s)$,
\begin{align}
  C_L\le\Lambda_g+\tfrac12\,\mathcal{B}_g.
\end{align}
\end{lem}
\begin{proof}
Integrating
Eq.~\eqref{eq:EPcurved} with $\mathcal E_P'=\operatorname{Tr}(PH')$ gives
\begin{align*}
  C_L
  &=\tfrac12\big[\operatorname{Tr}(P(0)H'(0))-\operatorname{Tr}(P(1)H'(1))\big] \\
  &\quad+\tfrac12\int_0^1\operatorname{Tr}(PH'')\,ds
  \;\le\;\Lambda_g+\tfrac12\,\mathcal{B}_g,
\end{align*}
bounding each trace by the extreme eigenvalue sums.
\end{proof}
Lemma~\ref{lem:speed} does not use affinity either, so the first inequality of
Theorem~\ref{thm:L} and the bound $L_P\le[C_L\,I_p(\Dstar)]^{1/2}$ of
Theorem~\ref{thm:rates} continue to hold, with
$C_L\le\Lambda_g+\mathcal{B}_g/2$ from Lemma~\ref{lem:budgetcurved} in place of
$C_L\le\Lambda_g$. The scaling along a family needs $C_L$ bounded over its
members, which for an affine family followed from the single bound on $\Lambda$ in
Definition~\ref{def:family}.
\begin{defn}\label{def:familycurved}
A family is in \emph{width class} $p$ if there exist constants $\Gamma$ and $C_W$
such that every member satisfies Definition~\ref{def:width} with this same pair,
and if single constants bound $\sup_s\|H'(s)\|$ and $\mathcal{B}_g$ over the family.
\end{defn}
For an affine family, $\mathcal{B}_g=0$ and $\sup_s\|H'(s)\|=\Lambda$, so
Definition~\ref{def:familycurved} generalizes Definition~\ref{def:family} to
non-affine families. A family in width class $p$ also has
$\Lambda_g$ bounded, because
$|\Lambda_g|\le g\sup_s\|H'(s)\|$, so $C_L$ stays bounded along it and the
non-affine dependence on $s$ does not change the scaling in $\Dstar$.

Unlike the length bound, the total curvature does not extend this way.
Theorem~\ref{thm:K} uses affinity through $H''=0$ in Eq.~\eqref{eq:sylvester}.
For non-affine $H(s)$, Appendix~\ref{app:winding} gives a family with the gap
fixed at $\Delta=2$ and with $\sup_s\|H'(s)\|$ and $\mathcal{B}_g$ bounded, whose
total curvature is unbounded.

\subsection{A constant bound for \texorpdfstring{$L$}{L} and \texorpdfstring{$K$}{K}?}
\label{sec:bounded}

So far we bounded the path length and the total curvature from the gap
profile, which gives logarithmic scaling at best. In many practical cases
these quantities are instead bounded by a constant, as in Secs.~\ref{sec:grover}
and~\ref{sec:xxz}, where the computed lengths stay near $\pi/2$ and $\pi/4$. The reason is geometric rather than spectral. The clearest example is a ground
state that stays inside a two-dimensional subspace that does not move with $s$,
while the rest of the spectrum is arbitrary. The following proposition makes
this precise.

\begin{prop}
Suppose $H(s)$ is affine as in Eq.~\eqref{eq:Hs} with $g=1$. Suppose also that a
two-dimensional subspace $\mathcal V$, independent of $s$, is invariant under
$H(s)$ and contains $|\Phi_0(s)\rangle$ for every $s\in[0,1]$. Then,
\begin{align}
  L\le\frac{\pi}{2},\qquad K=0.
  \label{eq:twolevel}
\end{align}
\label{prop:twolevel}
\end{prop}
\begin{proof}
On $\mathcal V$, up to a multiple of the identity, $H(s)$ acts as
$\tfrac12\mathbf b(s)\cdot\boldsymbol\sigma$ with
$\mathbf b(s)=\mathbf b_0+s\,\mathbf b_1$, where $\mathbf b_0,\mathbf b_1\in\mathbb R^3$, and $\mathbf b(s)\neq0$ because the ground
state is nondegenerate. Rotate the Bloch frame so that $\mathbf b_0$ and
$\mathbf b_1$ lie in the $xz$ plane. In the corresponding basis $H(s)$ is real
symmetric, so its ground state is real and stays in a fixed real plane of
$\mathcal V$. The unit vectors of that plane trace a great circle of the state
space, a geodesic, so $\kappa\equiv0$ and $K=0$ by Eq.~\eqref{eq:kappa}. If instead
the ground state does not move, $K=0$ trivially.

The ground state is fixed by the direction of $\mathbf b$. Its angular velocity
$(\mathbf b_0\times\mathbf b_1)/|\mathbf b|^2$ never changes direction, so $\mathbf b$
turns monotonically. It turns by less than $\pi$ in total, since $\mathbf b$ traces a
segment that misses the origin. The state turns by half as much, so $L$ is half the angle between
$\mathbf b(0)$ and $\mathbf b(1)$, less than $\pi/2$.
\end{proof}
Here $L$ approaches
$\pi/2$ as the initial and final states become orthogonal
[Fig.~\ref{fig:mechanism}(d)]. The Grover family of Sec.~\ref{sec:grover}
satisfies the hypothesis exactly, its ground state staying in the plane
spanned by the initial and the marked state, and the XXZ family of
Sec.~\ref{sec:xxz} comes close to lying in a fixed plane near its gap minimum.

The ground state of many avoided crossings stays in nearly the same plane near
their gap minimum, so one may expect $L$ and $K$ to stay bounded in such
systems. The $\sqrt{\log\Dstar^{-1}}$ and $\log\Dstar^{-1}$
scalings are reached instead when the ground state keeps mixing with a new
excited direction as $s$ increases, as in the family of
Appendix~\ref{app:sharpness}.

\subsection{Adiabatic evolution time}
\label{sec:runtime}

Finally, we consider how the width class and the bounds on the path geometry
act on the adiabatic evolution time $T$. The evolution time is strongly affected
by the schedule, that is, by how $s$ is changed with respect to the time $t$. A
discussion of the evolution time therefore comes with a specific choice of
schedule.
More precisely, take a schedule $s(\tau)$, a twice continuously
differentiable and nondecreasing function of $\tau=t/T$ with $s(0)=0$ and
$s(1)=1$, and let $|\psi(t)\rangle$ solve
$i\,\partial_t|\psi(t)\rangle=H(s(t/T))\,|\psi(t)\rangle$ from a normalized
$|\psi(0)\rangle\in\Ran P(0)$. With
$F_P(T)=\langle\psi(T)|P(1)|\psi(T)\rangle$ the population of the ground
subspace at the end of the evolution, the adiabatic evolution error $1-F_P(T)$
is at most $(\mathcal{C}[s]/T)^2$ with~\cite{Jansen2007}
\begin{align}
  \mathcal{C}[s]={}&\sqrt m\sum_{\tau\in\{0,1\}}\frac{\|\dot P\|}{\Delta(\tau)}
  +\sqrt m\int_0^1\frac{\|Q\ddot PP\|}{\Delta(\tau)}\,d\tau\nonumber\\
  &+\sqrt m\int_0^1\frac{\|\dot P\|^2}{\Delta(\tau)}\,d\tau
  +2m\int_0^1\frac{\|\dot H\|\,\|\dot P\|}{\Delta^2(\tau)}\,d\tau,
  \label{eq:cost}
\end{align}
where a dot is $\partial_\tau$, $\Delta(\tau)=\Delta(s(\tau))$, and $m$ bounds
the number of distinct eigenvalues inside the ground subspace along the path.

The first schedule we consider is the linear schedule, the standard choice.
It advances $s$ uniformly in time. Applying the
width class to this schedule gives the following corollary.
\begin{cor}\label{cor:runtime}
For twice continuously differentiable $H(s)$ whose gap profile is in width
class $p$, let $s_{\rm lin}(\tau)=\tau$ be the linear schedule. Write $\Lambda^{(1)}=\sup_s\|H'(s)\|$ and
$\Lambda^{(2)}=\sup_s\|H''(s)\|$, and let
\begin{align}
 \mathcal J_q=\frac1{\Gamma^q}
 \Big[1+\frac{qp}{qp-1}\,C_W\Big(\Big(\frac{\Gamma}{\Dstar}\Big)^{q-\frac1p}-1\Big)\Big].
 \label{eq:Jq}
\end{align}
Then $\mathcal C[s_{\rm lin}]\le C_T$ with
\begin{align}
 C_T={}&m\Lambda^{(1)}\Big(\frac1{\Delta^2(0)}+\frac1{\Delta^2(1)}\Big)
 \nonumber\\
 &+m\Lambda^{(2)}\,\mathcal J_2
 +7m^{3/2}(\Lambda^{(1)})^2\,\mathcal J_3.
 \label{eq:runtime}
\end{align}
\end{cor}
\begin{proof}
At $\tau=s$, bounding the two norms in Eq.~\eqref{eq:cost} by derivatives of
the Hamiltonian with the estimates of Jansen, Ruskai, and
Seiler~\cite{Jansen2007} gives
\begin{align*}
 \mathcal C[s_{\rm lin}]\le{}&m\Lambda^{(1)}\Big(\frac1{\Delta^2(0)}+\frac1{\Delta^2(1)}\Big)\\
 &+m\Lambda^{(2)}\int_0^1\frac{ds}{\Delta^2(s)}
 +7m^{3/2}(\Lambda^{(1)})^2\int_0^1\frac{ds}{\Delta^3(s)}.
\end{align*}
Writing $1/\Delta^q=\int_\Delta^\infty q\,\gamma^{-q-1}\,d\gamma$ and exchanging
the integrals as in Lemma~\ref{lem:layer}(i),
\begin{align*}
  \int_0^1\frac{ds}{\Delta^q(s)}
  =\frac{1}{\Gamma^q}
   +q\int_{\Dstar}^{\Gamma}\frac{\mubd(\gamma)}{\gamma^{q+1}}\,d\gamma .
\end{align*}
Inserting
Definition~\ref{def:width} then gives
$\int_0^1ds/\Delta^q(s)\le\mathcal J_q$ at $q=2$ and $q=3$, proving
Eq.~\eqref{eq:runtime}.
\end{proof}
Because
\begin{align}
  \mathcal J_q=\mathcal O\big(\Dstar^{-(q-\frac1p)}\big)
  \quad \textrm{as} \quad \Dstar\to0
  \label{eq:Jqscaling}
\end{align}
for $qp>1$, the last term of Eq.~\eqref{eq:runtime} sets the scaling of $C_T$.
Consider a family in width class $p$ as in
Definition~\ref{def:familycurved} along which $\Lambda^{(2)}$ stays bounded.
With the linear schedule, a fixed target population is then reached in
$T=\mathcal O(\Dstar^{-(3-1/p)})$, which is $\mathcal O(\Dstar^{-2})$ at
$p=1$. This improves on the worst case $T=\mathcal O(\Dstar^{-3})$ obtained
from the minimum gap alone~\cite{Jansen2007}. At $p=1$ with $g=1$,
Corollary~\ref{cor:runtime} recovers the linear-schedule scaling of Guo and
An~\cite{Guo2025}, and $\Dstar^{-2}$ is the scaling commonly
expected~\cite{Jansen2007} and observed~\cite{mchanCGS}.

Next, we consider the constant geometric speed (CGS)
schedule~\cite{mchanCGS}. In contrast to the linear schedule, which advances
$s$ uniformly in time, the CGS schedule advances the arc length of the ground
state path uniformly, $l=L_P\tau$. Applying the width class to
this schedule gives the following corollary.
\begin{cor}\label{cor:runtimecgs}
Let $H(s)$ be affine as in Eq.~\eqref{eq:Hs} with $g=1$, let $P$ be
nonconstant, and let the gap profile be in width class $p$. Write
$\Lambda=\|H'\|$ and let $\Omega$ be the difference between the largest and the
smallest eigenvalues of $H'$. Write $l(s)=\int_0^sv_P(r)\,dr$, so that $l(1)$
is the path length $L_P$ of Eq.~\eqref{eq:LPdef}. Let $s_c$ be the
schedule defined by $l(s_c(\tau))=L_P\tau$, which is well defined by
Lemma~\ref{lem:neverstops}. Then
\begin{align}
 \mathcal C[s_c]\le L_P\Big[\frac2{\Dstar}
 +\Big(\frac52\Omega+2\Lambda\Big)\mathcal J_2\Big].
 \label{eq:runtimecgs}
\end{align}
\end{cor}
\begin{proof}
At $g=1$ the ground subspace holds one eigenvalue, so $m=1$. In the gauge
$\langle\Phi_0|\dot\Phi_0\rangle=0$ the two norms in Eq.~\eqref{eq:cost} are
$\|\dot P\|=v$ and $\|Q\ddot PP\|=[\dot v^2+v^2\|Q\dot u\|^2]^{1/2}$, where
$v=\partial_\tau l$ is the speed in the scaled time and
$u=\partial_s\Phizeros/v_P$ is the unit tangent of
Sec.~\ref{sec:curvature}. The second norm follows from
$Q\,\partial_\tau^2|\Phi_0\rangle=\dot v\,u+v\,Q\dot u$ and
$\operatorname{Re}\langle u|Q\dot u\rangle=0$, as in the proof of
Theorem~\ref{thm:K}.

Under $s_c$ the speed is the constant $L_P$, so $\dot v=0$ and
$d\tau=v_P\,ds/L_P$. Each term of Eq.~\eqref{eq:cost} then carries a single
factor $L_P$,
\begin{align*}
 \mathcal C[s_c]={}&L_P\Big(\frac1{\Delta(0)}+\frac1{\Delta(1)}\Big)
 +L_P\int_0^1\frac{\|Qu'\|_{\HS}}{\Delta(s)}\,ds\\
 &+L_P\int_0^1\frac{v_P(s)}{\Delta(s)}\,ds
 +2L_P\Lambda\int_0^1\frac{ds}{\Delta^2(s)},
\end{align*}
where $\dot u=(L_P/v_P)\,u'$ and $\partial_\tau s_c=L_P/v_P$, so the change of
variables cancels the factor $v_P$ in the second and the fourth term.
Equation~\eqref{eq:turningpointwise} with $\|Y\|_{\HS}\le\Omega$, for the $Y$
of Eq.~\eqref{eq:Ydef}, bounds the second integrand by $2\Omega/\Delta^2$. For the third, expanding
$\partial_s\Phizeros$ in the eigenbasis and centering $H'$ at the midpoint $c$
of its spectrum gives
\begin{align*}
 v_P^2&=\sum_{j\ge1}\frac{|\langle\Phi_j|H'|\Phi_0\rangle|^2}{(E_j-E_0)^2}\\
 &\le\frac{\|Q(H'-c)\Phizeros\|^2}{\Delta^2(s)}\le\frac{\Omega^2}{4\Delta^2(s)},
\end{align*}
The shift by $c$ is free because $Q\Phizeros=0$, and $\|H'-c\|=\Omega/2$, so
that integrand is at most $\Omega/(2\Delta^2)$. The endpoint gaps are at
least $\Dstar$, and $\int_0^1ds/\Delta^2(s)\le\mathcal J_2$ as in the proof of
Corollary~\ref{cor:runtime}. Collecting the four terms proves
Eq.~\eqref{eq:runtimecgs}.
\end{proof}
Therefore, the CGS schedule gives $T=\mathcal{O}(\Dstar^{-(2-1/p)})$ for a width class $p$ if $L_P$ is
bounded by a constant independent of $\Dstar$. So, for a system with bounded adiabatic path length,
one order in $\Dstar^{-1}$ is gained by using the CGS schedule compared to the linear schedule
for any width class.
Even for the cases with $L_P$ unbounded, combining
Theorem~\ref{thm:rates} with Corollary~\ref{cor:runtimecgs} gives
$T=\mathcal{O}(\Dstar^{-1}[\log\Dstar^{-1}]^{1/2})$ for $p=1$ and
$T=\mathcal{O}(\Dstar^{-(5p-3)/(2p)})$ for $p>1$. This lowers the exponent by $(p+1)/(2p)$ compared to the
linear schedule for $p>1$. At $p=1$ the exponent falls from $2$ to $1$, with the
additional $[\log\Dstar^{-1}]^{1/2}$ factor.

Corollary~\ref{cor:runtimecgs} generalizes the result of our earlier
work~\cite{mchanCGS}. There, the one-order improvement is proved only for a
constantly bounded $L_P$ and between the bounds without the width class
information ($p\to\infty$), so from
$\Dstar^{-3}$ to $\Dstar^{-2}$. Here, we prove that one order for any width class
in the $L_P$ bounded case. And even $T=\mathcal{O}(\Dstar^{-1})$ itself is proved at $p=1$, which accounts for the
observations from the numerical tests there. Moreover, Corollary~\ref{cor:runtimecgs}
also covers cases with unbounded $L_P$, verifying the effectiveness of the CGS schedule even without the bounded $L_P$ assumption.

\section{Applications}
\label{sec:applications}

We illustrate the bounds on systems in three different categories. For each family we
give the pair $(\Gamma,C_W)$ that certifies width class $p=1$ for every member. We
then give the coefficient $C_L$ and compare Theorem~\ref{thm:rates} with the computed
$L$. The curvature is compared for the molecular families, where we also give $C_K$ and
Corollary~\ref{cor:Krates}. Every family here has a nondegenerate ground state, so $g=1$.

Here we describe how the bounds in the figures are evaluated. First, we use the computed
coefficients
$C_L$ and $C_K$. The conventional bounds are then
computed from \cref{eq:crudepathlengthbound} for $L$ and
\cref{eq:crudecurvaturebound} for $K$, and ours from Theorem~\ref{thm:rates} for
$L$ and Corollary~\ref{cor:Krates} for $K$. The two bounds coincide for
$\Dstar\ge\gamma_c$, and the shaded region between them is the improvement due to
the width class.
Second, for each family we use one pair $(\Gamma,C_W)$ and the largest values of $C_L$
and $C_K$ over its members, so one line in each panel serves the whole family.

\subsection{Adiabatic Grover search}
\label{sec:grover}
The adiabatic Grover search~\cite{Grover1997,Roland2002} drives the uniform
superposition $|\Phi_{\mathrm{in}}\rangle=|+\rangle^{\otimes n}$ toward a marked computational basis
state $|w\rangle$. The Hamiltonian is
$H(s)=(1-s)(\mathbb I-|\Phi_{\mathrm{in}}\rangle\langle\Phi_{\mathrm{in}}|)+s(\mathbb I-|w\rangle\langle w|)$,
with $\langle w|\Phi_{\mathrm{in}}\rangle=1/\sqrt N$ and $N=2^n$. The evolution stays in the
two-dimensional space spanned by $|w\rangle$ and $|\Phi_{\mathrm{in}}\rangle$, where
\begin{align}
  \Delta(s)=\sqrt{1-4(1-1/N)\,s(1-s)} .
  \label{eq:grovergap}
\end{align}
The gap minimum $\Delta_*=1/\sqrt N$ sits at $s_*=1/2$ and is lowered by
enlarging $N$ [Fig.~\ref{fig:apps}(a)].
With $\omega=\sqrt{1-1/N}$, the gap is
$\Delta(s)=\smash{\sqrt{\Dstar^2+4\omega^2(s-1/2)^2}}$, which is bounded below
linearly, Eq.~\eqref{eq:minorant} with $p=1$ and this $\omega$. Lemma~\ref{lem:minorant}
therefore gives width class $p=1$ with $\Gamma=1$ and $C_W=1/\omega\le1.04$ for
$N\ge16$, and Lemma~\ref{lem:budget} gives $C_L=\omega^2$. None of these constants
grows as $\Dstar$ shrinks.

The estimate
\cref{eq:crudepathlengthbound} gives $(\omega^2/\Dstar)^{1/2}$ and grows as a
power of $1/\Dstar$. Figure~\ref{fig:apps}(b) compares the computed $L$ with
\cref{eq:crudepathlengthbound} and with Theorem~\ref{thm:rates}, which is evaluated
from the family constants $\Gamma=1$, $C_W=1.04$ and $C_L=1$.
Because the ground states live in a fixed two-dimensional space that is
invariant under $H(s)$, Proposition~\ref{prop:twolevel} bounds $L$ by $\pi/2$, and
the computed $L$ approaches that value, far below the bound of
Theorem~\ref{thm:rates}.

\begin{figure}
  \centering
  \includegraphics[width=\columnwidth]{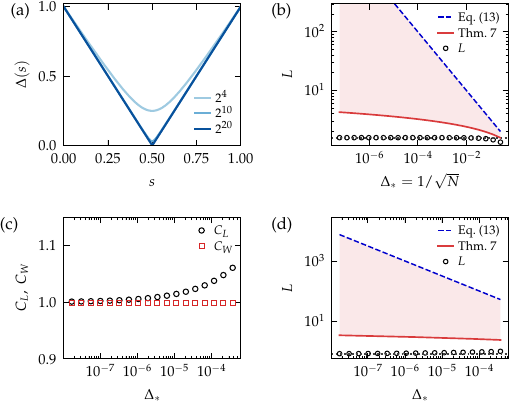}
  \caption{The adiabatic Grover search (a,b) and the XXZ chain (c,d). (a) The gap for
  $N=2^4,2^{10},2^{20}$. (b) $L$ versus $\Delta_*=1/\sqrt N$, with the conventional bound
  \cref{eq:crudepathlengthbound} (blue) and Theorem~\ref{thm:rates} (red), and the computed $L$ (black). The dotted line is $\pi/2$. (c) $C_L$ (black) and
  $C_W$ (red) versus $\Delta_*$. (d) The same as (b) for the XXZ family, with the
  dotted line at $\pi/4$. In (b) and (d), the shaded area fills the region between the two bounds.}
  \label{fig:apps}
\end{figure}

\subsection{XXZ spin chain}
\label{sec:xxz}
We next consider an open XXZ chain~\cite{YangYang1966} of $n=10$ spins with
\begin{align}
  H(s)={}&J_z\sum_i S^z_iS^z_{i+1}
       +s\,J_{xy}\sum_i\big(S^x_iS^x_{i+1}+S^y_iS^y_{i+1}\big) \nonumber\\
       &+(1-s)\,h_0\sum_i(-1)^iS^z_i.
  \label{eq:xxzpath}
\end{align}
The couplings are in the easy-axis regime $J_z>J_{xy}$. At $s=0$ the transverse
coupling is off and $H$ is diagonal, and the ground state is the single Néel product
selected by the staggered field. As $s$ runs to $1$ the transverse coupling turns on and the
staggered field turns off. The target at $s=1$ is the ground state of the pure XXZ
chain, dominated by the two Néel orderings, with transverse corrections that vanish as
$J_{xy}\to0$. The gap minimum sits at the endpoint $s=1$, where the staggered field has vanished and
the global spin flip $\prod_i\sigma^x_i$ is an exact symmetry that exchanges the two
Néel orderings, whose combinations are close to the two lowest states there.

We compute the family at $J_z=2$ and $h_0=0.4$ for $18$ values of $J_{xy}$ from $0.21$
down to $0.028$, which drive $\Dstar$ from $3.8\times10^{-4}$ to $1.7\times10^{-8}$.
The computed gap profiles give
$\Gamma=1.8$ and $C_W=1.00$ at $p=1$, while $C_L\le1.061$ and
neither $C_L$ nor $C_W$ grows as the gap closes
[Fig.~\ref{fig:apps}(c)]. Figure~\ref{fig:apps}(d) reports the result. The bound of
Theorem~\ref{thm:rates} runs from $2.4$ to $3.4$ over these four decades of $\Dstar$
and stays above every computed $L$ by a factor between $2.4$ and $4.1$. The
conventional bound \cref{eq:crudepathlengthbound} instead reaches $7.9\times10^3$,
above the computed $L$ by a factor of $9.6\times10^3$.

The computed $L$ itself does not grow, staying between $0.82$ and $0.98$. Along the
whole path the ground state keeps at least $95\%$ of its weight in the plane spanned by
the two Néel orderings. Because the gap minimum sits at $s=1$, only half of the avoided
crossing is traversed. This accounts for the value near $\pi/4$, half of the $\pi/2$ of
Proposition~\ref{prop:twolevel}.

\subsection{Molecular electronic Hamiltonians}
\label{sec:chemistry}
Finally we turn to the electronic structure of molecules. The Hamiltonian $H(s)=H_i+s(H_f-H_i)$
runs from a mean-field $H_i$ to the full interacting $H_f$, so $H_0=H_i$ and $H'=H_f-H_i$. We
consider the two molecular systems of our earlier work~\cite{mchanCGS}, the nitrogen molecule over a
range of bond lengths and a strongly correlated bioinorganic [2Fe-2S]
cluster~\cite{Beinert1997,Johnson2005,Reiher2017,Sharma2014}. For both systems, $H_i$ is
the Kohn--Sham Hamiltonian of a scalar-relativistic BP86 density-functional
calculation~\cite{Becke1988,Perdew1986}. For N$_2$, we use the
STO-3G basis~\cite{Stewart1970} and consider all $14$ electrons in its $10$ orbitals. For
[2Fe-2S], we use the TZP-DKH basis~\cite{Jorge2009} and an active space of $14$ electrons in
$12$ Kohn--Sham orbitals~\cite{Lee2023}, the $10$ Fe $3d$ orbitals and the two $3p$ orbitals
of the bridging sulfurs.
Each member of the [2Fe-2S] family starts from a different Slater determinant, and we
shift the Kohn--Sham orbital energies in $H_i$ so that this determinant has the lowest energy.
Since $H(s)$ conserves some symmetries (the total spin for both systems and the $D_{2h}$
point-group symmetry for N$_2$), the adiabatic path starting from the ground state of $H_i$ stays
in the corresponding symmetry sector, where we compute $\Delta(s)$.

For the $37$ bond lengths of N$_2$, from $0.4$ to $4.0$ \AA{} in steps of
$0.1$ \AA{}, the minimum gap runs from $1.45$ down to $1.8\times10^{-4}$ hartree. The
profiles give $\Gamma=1.51$ and $C_W=8.6$ at $p=1$, with $C_L\le0.535$ and
$C_K\le1.13$ hartree. For the $83$ determinant starts of [2Fe-2S], whose minimum gap runs
from $4.7\times10^{-2}$ down to $3.5\times10^{-5}$ hartree, we have $\Gamma=0.50$
and $C_W=20.7$, with $C_L\le2.422$ and $C_K\le1.37$ hartree.

Figure~\ref{fig:chemistry} plots $L$ and $K$ against $\Dstar$ together with the
resulting bounds. The bound of Theorem~\ref{thm:rates} lies above every computed $L$, by a factor between
$1.9$ and $7.1$ for N$_2$ [Fig.~\ref{fig:chemistry}(a)] and between $2.4$ and $7.3$
for [2Fe-2S] [Fig.~\ref{fig:chemistry}(c)]. The bound of Corollary~\ref{cor:Krates} lies above every
computed $K$, by $4.9$ to $22.8$ for N$_2$ [Fig.~\ref{fig:chemistry}(b)] and by $12.8$ to
$95.4$ for [2Fe-2S] [Fig.~\ref{fig:chemistry}(d)]. At the smallest gap of N$_2$ the
curvature bound improves on \cref{eq:crudecurvaturebound} by a factor of $122$ and the
length bound on \cref{eq:crudepathlengthbound} by $11$. For [2Fe-2S] at its smallest gap, the curvature bound improves by a
factor of $92$ and the length bound by $9.6$.

\begin{figure}
  \centering
  \includegraphics[width=\columnwidth]{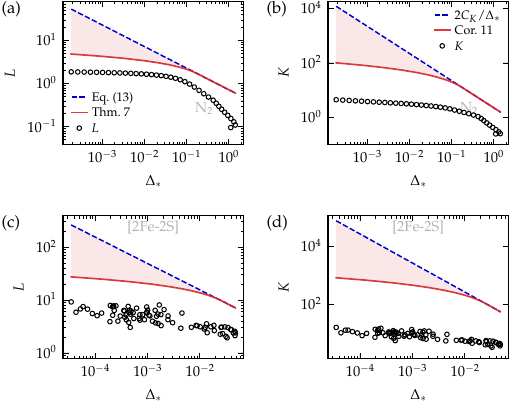}
  \caption{Adiabatic path length $L$ and total curvature $K$ versus $\Delta_*$ for the
  N$_2$ bond-length family (a,b) and the [2Fe-2S] determinant family (c,d).
  Open circles are the computed values. Solid lines are the bounds of
  Theorem~\ref{thm:rates} and Corollary~\ref{cor:Krates}, dashed lines the conventional
  bounds \cref{eq:crudepathlengthbound} and \cref{eq:crudecurvaturebound}, with
  the same coefficients. The shaded area fills the region between the two bounds.}
  \label{fig:chemistry}
\end{figure}

\section{Conclusion}
\label{sec:conclusion}
The conventional bounds on the adiabatic path length $L$ and the total curvature $K$ are
too conservative. In practice, as shown in Sec.~\ref{sec:applications}, the deviation
between the bounds and the computed values becomes as large as orders of magnitude as the
minimum gap $\Dstar$ decreases to $0$. The reason for this discrepancy is that the bounds
use only $\Dstar$, how deep the gap profile is, to bound the gap integral
$\int_0^1ds/\Delta$. Here, we have shown that introducing the width class, a measure of how narrow the
small-gap region is, improves the bounds up to polylogarithmic ones for a typical
avoided crossing.

More specifically, width class $p=1$ makes the gap integral logarithmic, so the length grows
at most as $[\log(\Dstar^{-1})]^{1/2}$ and the curvature at most as $\log(\Dstar^{-1})$.
For $p>1$ the bounds are powers of $1/\Dstar$ with exponents $(p-1)/(2p)$ and $(p-1)/p$,
below the $1/2$ and the $1$ of the conventional bounds \cref{eq:crudepathlengthbound} and
\cref{eq:crudecurvaturebound}. The gap scaling of the path length is proved to be tight
for every integer $p\geq1$, while the tightness of the curvature scaling is proved only
for $p=1$ (Appendix~\ref{app:sharpness}). Moreover, independently of the bounds on the path length and the curvature, the width class
also improves the runtime bound of the adiabatic theorem for the linear schedule, from
$\mathcal O(\Dstar^{-3})$ to $\mathcal O(\Dstar^{-(3-1/p)})$, which is $\mathcal O(\Dstar^{-2})$ at $p=1$
(Corollary~\ref{cor:runtime}). For the constant geometric speed schedule the
width class and the length bound of Theorem~\ref{thm:rates} together give
$\mathcal O(\Dstar^{-1}[\log\Dstar^{-1}]^{1/2})$ at $p=1$
(Corollary~\ref{cor:runtimecgs}).

We evaluated the bounds on the adiabatic Grover
Hamiltonian, an XXZ chain, and two families of molecular Hamiltonians, all four of which
turn out to be in width class $p=1$. The resulting bounds lie above every computed $L$,
and above every computed $K$ for the molecular families, and well below the conventional
bound at small $\Dstar$.

Three directions remain open.
The first is a clear understanding of the mechanism that bounds the adiabatic path
length and the total curvature by a constant.
Though the width class improves the bounds up to polylogarithmic ones, the computed $L$
and $K$ appear to be bounded by a constant as $\Dstar\to0$ (Sec.~\ref{sec:xxz} and the
N$_2$ family of Sec.~\ref{sec:chemistry}). Proposition~\ref{prop:twolevel} partially answers this
question. Extending it to more complicated systems may resolve it. The second
is the tightness of the curvature scaling, which is proved at $p=1$ in this paper and open for
$p>1$. The third is the design of adiabatic algorithms based on the path geometry.
The constant geometric speed schedule~\cite{mchanCGS} uses the path geometry so
that the prior spectral knowledge needed to run the schedule is only a lower bound on the gap
rather than the whole gap function, and Corollary~\ref{cor:runtimecgs} gives the
schedule a performance guarantee. The path geometry may reduce the prior spectral
knowledge of other adiabatic
algorithms in the same way, and constructing such algorithms with performance
guarantees like Corollary~\ref{cor:runtimecgs} is left for future work.

\begin{acknowledgments}
We used resources of
the Center for Advanced Computation at Korea
Institute for Advanced Study
and the National Energy Research Scientific Computing Center (NERSC), a U.S. Department of Energy
Office of Science User Facility operated under
Contract No. DE-AC02-05CH11231.
SC was supported by a KIAS Individual Grant (No. CG090601) at Korea Institute for Advanced Study
and by Quantum Simulator Development Project for Materials Innovation through the National
Research Foundation of Korea (NRF) funded by the Korean government
(Ministry of Science and ICT (MSIT)) (No. NRF-2023M3K5A1094813).
MH was supported by a KIAS Individual Grant (No. CG091302) and
by Institute of Information \& Communications Technology Planning \& Evaluation (IITP) grant
funded by the Korea government (MSIT) (No. 2022-0-01026).
\end{acknowledgments}

\bibliography{refs}

\appendix

\makeatletter
\@addtoreset{thm}{section}
\makeatother
\renewcommand{\thethm}{\thesection\arabic{thm}}
\renewcommand{\theHthm}{appendix.\thesection.\arabic{thm}}
\renewcommand{\thesubsection}{\thesection.\arabic{subsection}}

\section{Families saturating the bounds}
\label{app:sharpness}

Two families saturate the bounds of the main text. The first reaches
$L_P=\Theta(\sqrt{\log\Dstar^{-1}})$ and $K_P=\Theta(\log\Dstar^{-1})$
simultaneously at width class one, the scalings of Theorem~\ref{thm:rates} and
Corollary~\ref{cor:Krates}, through a single avoided crossing. The second
saturates the power law of Theorem~\ref{thm:rates} at width class $p$ for each
fixed integer $p\geq2$. Both are tensor sums
$A\boxplus B=A\otimes\mathbb I_B+\mathbb I_A\otimes B$, where $\mathbb I_A$ and
$\mathbb I_B$ are the identities on the spaces of $A$ and $B$.
Lemma~\ref{lem:tensorsum} collects the properties of tensor sums that both
proofs use.

\begin{lem}\label{lem:tensorsum}
Let $A(s)$ and $B(s)$ be real symmetric and continuously differentiable matrix-valued functions on an interval,
each of dimension at least two, with nondegenerate ground states $|\phi_A\rangle$ and
$|\phi_B\rangle$ and positive gaps $\Delta_A(s)$ and $\Delta_B(s)$. Then $H(s):=A(s)\boxplus B(s)$ has the
nondegenerate ground state $|\Phi_0\rangle=|\phi_A\rangle\otimes|\phi_B\rangle$, whose energy is the
sum of the ground energies of $A(s)$ and $B(s)$. The gap of $H(s)$ is
$\Delta(s)=\min\{\Delta_A(s),\Delta_B(s)\}$. If $|\phi_A\rangle$
and $|\phi_B\rangle$ are chosen real normalized and continuously differentiable, then
\begin{align}
 v_P^2=\|\partial_s\phi_A\|^2+\|\partial_s\phi_B\|^2.
 \label{eq:tensorsum-speed}
\end{align}
\end{lem}
\begin{proof}
The eigenvalues of a tensor sum are the sums of those of its two factors,
which proves the assertions on the ground state, the ground energy and the gap.
The product rule gives
$\partial_s|\Phi_0\rangle=|\partial_s\phi_A\rangle\otimes|\phi_B\rangle
+|\phi_A\rangle\otimes|\partial_s\phi_B\rangle$.
Real normalized vectors satisfy $\langle\phi_A|\partial_s\phi_A\rangle=
\langle\phi_B|\partial_s\phi_B\rangle=0$, so the two terms are orthogonal and
$\langle\Phi_0|\partial_s\Phi_0\rangle=0$. With $P=|\Phi_0\rangle\langle\Phi_0|$,
Eq.~\eqref{eq:vpdef} then gives $v_P=\|\partial_s\Phi_0\|$, which proves
Eq.~\eqref{eq:tensorsum-speed}.
\end{proof}
With more than two factors the tensor sum is formed in the same way, for example
$A\boxplus B\boxplus C=A\otimes\mathbb I_B\otimes\mathbb I_C
+\mathbb I_A\otimes B\otimes\mathbb I_C+\mathbb I_A\otimes\mathbb I_B\otimes C$.
Since $A\boxplus B\boxplus C=(A\boxplus B)\boxplus C$,
Lemma~\ref{lem:tensorsum} applies factor by factor.

\subsection{Width class one}

Let $N$ be a positive integer and let $s\in[0,1]$.
Let $a_j=2^{-j}$ for $1\le j\le N$.
Let $b=2^{-20}$ and $\beta_N=b/\sqrt N$.
Let $\zeta_N=\sqrt{1-a_N^2}$.
Define

\begin{align}
 B_N(s)&=\frac12
 \begin{pmatrix}\zeta_Ns&a_N\\ a_N&-\zeta_Ns\end{pmatrix},
 \nonumber\\
 A_N(s)&=\begin{pmatrix}
 0&\beta_Ns&\cdots&\beta_Ns\\
 \beta_Ns&s+2a_1&&0\\
 \vdots&&\ddots&\\
 \beta_Ns&0&&s+2a_N
 \end{pmatrix},\nonumber\\
 H_N(s)&=A_N(s)\boxplus B_N(s).
 \label{eq:stair-H}
\end{align}

The factor $B_N$ is a single avoided crossing, and it alone sets the gap, with minimum
$a_N$ at $s=0$. The growth of the path length comes mainly from $A_N$. Write
$|0\rangle,\dots,|N\rangle$ for the coordinate basis of $A_N$, whose ground state is
$|0\rangle$ at $s=0$. Because the $j$-th diagonal entry of
$A_N$ is $s+2a_j$ and the first basis vector couples to the $j$-th with strength
$\beta_N s$, the ground state acquires a $j$-th component of size about
$\beta_N s/(s+2a_j)=\beta_N/(1+2a_j/s)$. That component vanishes for $s\ll a_j$ and is of
order $\beta_N$ near $s=a_j$. Because one such growth occurs at each $a_j=2^{-j}$, the
path length adds up to about $N\beta_N=b\sqrt N$. With $\Dstar=a_N=2^{-N}$ this
gives $L_P=\Theta(\sqrt{\log\Dstar^{-1}})$. The total curvature likewise collects one
contribution at each $a_j$, but of scale one, which gives
$K_P=\Theta(N)=\Theta(\log\Dstar^{-1})$.

\begin{lem}\label{lem:stair-gap}
Let $N$ be a positive integer and let $H_N(s)$ be defined by
Eq.~\eqref{eq:stair-H} for $s\in[0,1]$.
Then $H_N$ is real symmetric and affine with dimension $2(N+1)$ and a
nondegenerate ground state.
Its gap is in width class one with
\begin{align}
 \Delta(s)&=\sqrt{a_N^2+\zeta_N^2s^2},\qquad \Dstar=a_N,\nonumber\\
 \Gamma&=1,\qquad C_W=\frac2{\sqrt3}.
 \label{eq:stair-gap}
\end{align}
It satisfies $\|H_N'\|<2$ and $\Omega<4$.
Here $\Omega$ is the difference between the largest and smallest eigenvalues
of $H_N'$. The $H_N$ form a family of width class one in
Definition~\ref{def:family}.
\end{lem}
\begin{proof}
We write $\Delta_A$ and $\Delta_B$ for the gaps of the two factors $A_N$ and
$B_N$, and let $e_0\le e_1$ be the two lowest eigenvalues of $A_N(s)$, so
$\Delta_A=e_1-e_0$. Consider $\operatorname{diag}(s+2a_j)$, which is obtained by
deleting the first row and column of $A_N$. By the Cauchy interlacing
theorem~\cite{Bhatia1997}, $e_1$ is at least the smallest of these diagonal entries,
$s+2a_N$. The same theorem gives $e_0<0$ for $s>0$, because the
$2\times2$ block of $A_N$ on the first and any other coordinate has a negative
determinant, $-\beta_N^2s^2$. Hence
\begin{align}
 \Delta_A(s)&>s+2a_N>\Delta_B(s)\qquad(s>0),\nonumber\\
 \mubd(\gamma)&=\frac{\sqrt{\gamma^2-a_N^2}}{\zeta_N}
 \leq\frac2{\sqrt3}\frac{\gamma}{\Gamma}
 \qquad(a_N\le\gamma\le1),
\end{align}
where $s+2a_N>\Delta_B(s)$ holds because $\zeta_N\le1$ gives
$\Delta_B(s)\le\sqrt{a_N^2+s^2}\le a_N+s$. At $s=0$ the couplings vanish and
$\Delta_A(0)=2a_N>a_N=\Delta_B(0)$.
Lemma~\ref{lem:tensorsum} gives the gap of $H_N$ as $\Delta_B$ and proves
Eq.~\eqref{eq:stair-gap}.

Differentiating $A_N$ gives a constant matrix, the sum of
$\operatorname{diag}(0,1,\ldots,1)$ and a matrix whose only nonzero entries are
$\beta_N$ in the first row and column. The second has norm $\beta_N\sqrt N=b$, so
$\|A_N'\|\le1+b$. The eigenvalues of $B_N'$ are $\pm\zeta_N/2$, so the triangle inequality gives $\|H_N'\|\le\frac32+b<2$. Hence $\Omega<4$.
The constants $\Gamma$ and $C_W$ and the bound on $\|H_N'\|$ do not depend on $N$,
as Definition~\ref{def:family} requires.
\end{proof}

\begin{lem}\label{lem:stair-state}
Let $N$ be a positive integer and let $A_N(s)$ be defined in
Eq.~\eqref{eq:stair-H} for $0<s\le1$.
Let $e_0(s)$ be its lowest eigenvalue and let $y(s)=-e_0(s)/s$.
Let $|\phi_A(s)\rangle$ be its real normalized ground state with positive
first coordinate.
In the coordinate basis $|0\rangle,\ldots,|N\rangle$ it is
\begin{align}
 |\phi_A\rangle=\frac{|0\rangle-|z\rangle}{\sqrt{1+\|z\|^2}},\qquad
 z_j(s)=\frac{\beta_Ns}{s(1+y)+2a_j}.
 \label{eq:stair-state}
\end{align}
Here $z_0=0$ and $|z\rangle=\sum_{j=1}^N z_j|j\rangle$.
The scalar $y$ satisfies
\begin{align}
 y=\beta_N^2\sum_{j=1}^N \frac{s}{s(1+y)+2a_j}
 \label{eq:stair-y}
\end{align}
and
\begin{align}
 0\le y\le b^2,\qquad \|z\|\le b.
 \label{eq:stair-ybounds}
\end{align}
The ground state $|\phi_A\rangle$ satisfies
\begin{align}
 \partial_s|\phi_A\rangle&=-\frac{(\mathbb I-|\phi_A\rangle\langle\phi_A|)\,
 \partial_s|z\rangle}{\sqrt{1+\|z\|^2}},
 \label{eq:stair-dphi}\\
 \|(\mathbb I-|\phi_A\rangle\langle\phi_A|)x\|&\ge\frac{\|x\|}{\sqrt{1+\|z\|^2}}
 \label{eq:stair-proj}
\end{align}
for every vector $|x\rangle$ with $x_0=0$.
\end{lem}
\begin{proof}
The eigenvalue equation gives Eqs.~\eqref{eq:stair-state} and~\eqref{eq:stair-y}.
For $s>0$ the $2\times2$ block of $A_N$ on coordinates $0$ and $1$ has
determinant $-\beta_N^2s^2<0$, so the Cauchy interlacing theorem~\cite{Bhatia1997}
gives $e_0<0$.
Hence $y>0$ and each term of the sum in Eq.~\eqref{eq:stair-y} lies in $[0,1]$.
Hence $y\le N\beta_N^2=b^2$. Since $z_j$ is $\beta_N$ times the $j$-th term,
$0\le z_j\le\beta_N$ and $\|z\|\le\sqrt N\beta_N=b$.

Write $|w\rangle=|0\rangle-|z\rangle$, so that $|\phi_A\rangle=|w\rangle/\|w\|$
with $\|w\|^2=1+\|z\|^2$. The derivative of a normalized vector is
$\partial_s(|w\rangle/\|w\|)=(\mathbb I-|\phi_A\rangle\langle\phi_A|)
\partial_s|w\rangle/\|w\|$. With $\partial_s|w\rangle=-\partial_s|z\rangle$ this
gives Eq.~\eqref{eq:stair-dphi}. If $x_0=0$, then
$\langle\phi_A|x\rangle=-\langle z|x\rangle/\|w\|$. The Cauchy--Schwarz
inequality bounds $|\langle\phi_A|x\rangle|^2$ by $\|z\|^2\|x\|^2/(1+\|z\|^2)$, and
$\|(\mathbb I-|\phi_A\rangle\langle\phi_A|)x\|^2=\|x\|^2-|\langle\phi_A|x\rangle|^2$
gives Eq.~\eqref{eq:stair-proj}.
\end{proof}

\begin{thm}\label{thm:stair-length}
Let $N$ be a positive integer and let $H_N(s)$ be defined by
Eq.~\eqref{eq:stair-H} for $s\in[0,1]$.
Let $L_P$ be its ground-state path length and let
$c_b=b(1-b^2)/[(6+3b^2)(1+b^2)]$ with $b=2^{-20}$.
Then $c_b\sqrt N\leq L_P<\sqrt{2(N+2)}$. Along the family $\{H_N\}$,
$L_P=\Theta(\sqrt{\log\Dstar^{-1}})$ as $N\to\infty$.
\end{thm}
\begin{proof}
Since $z_0=0$, Eqs.~\eqref{eq:stair-dphi} and~\eqref{eq:stair-proj} with
$|x\rangle=|z'\rangle$ give $\|\phi_A'\|\ge\|z'\|/(1+\|z\|^2)\ge\|z'\|/(1+b^2)$,
the last step by Eq.~\eqref{eq:stair-ybounds}. Lemma~\ref{lem:tensorsum} gives $v_P\ge\|\phi_A'\|$.
Evaluating Eq.~\eqref{eq:stair-state} at $s=a_j$ and $s=2a_j$ with
$y\in[0,b^2]$ from Eq.~\eqref{eq:stair-ybounds} gives
\begin{align}
 z_j(2a_j)-z_j(a_j)\geq
 \beta_N\Big(\frac1{2+b^2}-\frac13\Big)
 =\frac{\beta_N(1-b^2)}{6+3b^2}.
 \label{eq:stair-shells}
\end{align}
Moreover,
\begin{align*}
  \int_0^1\|z'\|\,ds
  &\;\ge\;\sum_{j=1}^{N}\int_{a_j}^{2a_j}\|z'\|\,ds\\
  &\;\ge\;\sum_{j=1}^{N}\big[z_j(2a_j)-z_j(a_j)\big].
\end{align*}
With Eq.~\eqref{eq:stair-shells} and $N\beta_N=b\sqrt N$, the right-hand side is
at least $b\sqrt N(1-b^2)/(6+3b^2)$. Since $v_P\ge\|\phi_A'\|\ge\|z'\|/(1+b^2)$, the
path length $L_P=\int_0^1v_P\,ds$ is at least $c_b\sqrt N$.

At $g=1$, $\Lambda_1=\Omega/2$, so Lemma~\ref{lem:stair-gap} gives $\Lambda_1<2$.
Theorem~\ref{thm:L} then gives $L_P^2\le\Lambda_1(N+2)<2(N+2)$ from
\begin{align}
 \int_0^1\frac{ds}{\Delta(s)}
 =\frac{\arsinh(\zeta_N/a_N)}{\zeta_N}\le N+2.
 \label{eq:stair-gapint}
\end{align}
With $\Dstar=a_N=2^{-N}$ from Lemma~\ref{lem:stair-gap}, the two bounds on $L_P$
give $L_P=\Theta(\sqrt{\log\Dstar^{-1}})$.
\end{proof}

\begin{lem}\label{lem:stair-unitcompare}
Let $N$ be a positive integer and consider $0<s\le1$. Let $\xi=\log s$, so that
$\partial_\xi=s\partial_s$, and $|r\rangle:=\beta_N^{-1}\partial_\xi|z\rangle$
with $|z\rangle$ and $|\phi_A\rangle$ of Lemma~\ref{lem:stair-state}. Define
\begin{align*}
 f_0&=0,\qquad f_j=\frac{2a_js}{(s+2a_j)^2}\quad(1\le j\le N),\\
 |e\rangle&=-\frac{|f\rangle}{\|f\|}.
\end{align*}
Then $\|r-f\|\le3b^2\|f\|$ and $w_A:=\|\partial_\xi\phi_A\|>0$. Moreover, for every
integer $1\le k\le N$, the unit tangent $|u_A\rangle:=\partial_\xi|\phi_A\rangle/w_A$
satisfies
\begin{align}
 \|u_A(a_k)-e(a_k)\|\le2(3b^2+b)<6b.
 \label{eq:stair-unitcompare}
\end{align}
\end{lem}
\begin{proof}
By Lemma~\ref{lem:neverstops}, the ground state of $A_N$ has positive speed
because Eq.~\eqref{eq:stair-state} differs from $|\phi_A(0)\rangle=|0\rangle$
for $s>0$.
With $D_j=s(1+y)+2a_j$, differentiating Eq.~\eqref{eq:stair-y} gives
\begin{align*}
 0\le\partial_\xi y
 =\frac{\beta_N^2\sum_{j=1}^N 2a_js/D_j^2}
 {1+\beta_N^2\sum_{j=1}^N s^2/D_j^2}
 \le\beta_N^2\sum_{j=1}^N f_j.
\end{align*}
Differentiating $|z\rangle$ in Eq.~\eqref{eq:stair-state} gives $r_0=0=f_0$ and
$r_j=(2a_js-s^2\partial_\xi y)/D_j^2$ for $1\le j\le N$.
The first term of $r_j$ differs from $f_j$ by at most $2yf_j\le2b^2f_j$, since
$s+2a_j\le D_j\le(1+y)(s+2a_j)$ and $y\le b^2$. In each coordinate the second term
$s^2\partial_\xi y/D_j^2$ is at most $\partial_\xi y$, since $D_j\ge s$. Then, from
$\sum_{j=1}^N f_j\le\sqrt N\|f\|$ and $N\beta_N^2=b^2$, we have
$\|r-f\|\le2b^2\|f\|+\sqrt N\,\partial_\xi y\le3b^2\|f\|$.
Moreover, $\partial_\xi=s\partial_s$ and $\partial_\xi|z\rangle=\beta_N|r\rangle$
turn Eq.~\eqref{eq:stair-dphi} into
\begin{align}
 \frac{\sqrt{1+\|z\|^2}}{\beta_N}\partial_\xi|\phi_A\rangle
 &=-(\mathbb I-|\phi_A\rangle\langle\phi_A|)|r\rangle.
 \label{eq:stair-tangent}
\end{align}
For any nonzero vectors $|\psi_1\rangle$ and $|\psi_2\rangle$,
$\|\psi_1/\|\psi_1\|-\psi_2/\|\psi_2\|\|\le2\|\psi_1-\psi_2\|/\|\psi_2\|$.
Take $|\psi_1\rangle$ to be the vector in Eq.~\eqref{eq:stair-tangent} and
$|\psi_2\rangle=-|f\rangle$. Then $|\psi_1\rangle$ and $|\psi_2\rangle$ are normalized
to $|u_A\rangle$ and $|e\rangle$, respectively, and $\|\psi_2\|=\|f\|$. We also have
\begin{align*}
 |\psi_1\rangle-|\psi_2\rangle
 =(\mathbb I-|\phi_A\rangle\langle\phi_A|)(|f\rangle-|r\rangle)
 +|\phi_A\rangle\langle\phi_A|f\rangle.
\end{align*}
Since $f_0=0$ and $\|z\|\le b$ by Eq.~\eqref{eq:stair-ybounds},
$|\langle\phi_A|f\rangle|=\|f\|\,|\langle z|e\rangle|/\sqrt{1+\|z\|^2}\le b\|f\|$.
Hence $\|\psi_1-\psi_2\|\le\|r-f\|+|\langle\phi_A|f\rangle|\le(3b^2+b)\|f\|$, and
Eq.~\eqref{eq:stair-unitcompare} follows.
\end{proof}

\begin{lem}\label{lem:stair-fourshift}
Let $N$ be a positive integer and let $|f\rangle$ and $|e\rangle$ be as in
Lemma~\ref{lem:stair-unitcompare}. Then $2/9\le\|f(a_k)\|\le\sqrt3/2$ for every
integer $1\le k\le N$. For every integer $1\le k\le N-4$,
\begin{align}
 \|e(a_k)-e(a_{k+4})\|
 \ge\frac29-\frac9{64}>\frac1{30}.
 \label{eq:stair-fourshift}
\end{align}
\end{lem}
\begin{proof}
With $\rho=2a_j/a_k=2^{k-j+1}$ we have
$f_j(a_k)=\rho/(1+\rho)^2$. Then $f_j(a_k)\le1/4$ because $(1+\rho)^2\ge4\rho$, and
$f_j(a_k)\le2^{-|k-j+1|}$ because $(1+\rho)^2\ge\max\{1,\rho\}^2$. These give
\begin{align*}
 \sum_{j=1}^N f_j(a_k)&\le\sum_{m\in\mathbb Z}2^{-|m|}=3,\\
 \|f(a_k)\|^2&\le\tfrac14\sum_{j=1}^N f_j(a_k)\le\tfrac34,\\
 f_k(a_{k+4})&\le2^{-5}=\frac1{32}.
\end{align*}
Thus coordinate $k$ of $|f\rangle/\|f\|$ is at least $2/9$ at $a_k$ and at most
$9/64$ at $a_{k+4}$. The first bound holds since $f_k(a_k)=2/9$ and
$\|f(a_k)\|<1$, and the second bound comes from $f_k(a_{k+4})\le1/32$ and
$\|f(a_{k+4})\|\ge f_{k+4}(a_{k+4})=2/9$.
Since $\|e(a_k)-e(a_{k+4})\|$ is at least the absolute difference of the $k$-th
coordinates, Eq.~\eqref{eq:stair-fourshift} follows.
\end{proof}

\begin{lem}\label{lem:stair-productcompare}
Let $N\ge2^{24}$ be an integer, and take $|\phi_A\rangle$ from
Lemma~\ref{lem:stair-state} and $|e\rangle$ from Lemma~\ref{lem:stair-unitcompare}.
Choose the ground state $|\phi_B\rangle$ of $B_N$ in Eq.~\eqref{eq:stair-H} to be real and normalized with
positive second coordinate. Then for every integer $1\le k\le\lfloor N/2\rfloor$ the
unit tangent $u:=\partial_s|\Phi_0\rangle/v_P$ of the ground state
$|\Phi_0\rangle=|\phi_A\rangle\otimes|\phi_B\rangle$ of $H_N$ satisfies
\begin{align}
 \|u(a_k)-e(a_k)\otimes\phi_B^\infty\|<7b,\qquad
 |\phi_B^\infty\rangle:=(0,1)^{\mathsf T}.
 \label{eq:stair-productcompare}
\end{align}
\end{lem}
\begin{proof}
We use $\xi$ and $w_A$ and the vectors $|r\rangle$, $|f\rangle$ and $|u_A\rangle$ of Lemma~\ref{lem:stair-unitcompare}.
With $\theta=\arctan(a_N/(\zeta_Ns))$ the ground state of $B_N$ is
$|\phi_B\rangle=(-\sin\tfrac\theta2,\cos\tfrac\theta2)^{\mathsf T}$. Write
$w_B=\|\partial_\xi\phi_B\|$ and $|u_B\rangle=\partial_\xi|\phi_B\rangle/w_B$
where $w_B>0$.
Since $\partial_\xi=s\partial_s$ with $s>0$, $u$ is also the normalized
$\partial_\xi|\Phi_0\rangle$. The product rule and Eq.~\eqref{eq:tensorsum-speed} give
\begin{align*}
 u=\frac{w_A\,|u_A\rangle\otimes|\phi_B\rangle+w_B\,|\phi_A\rangle\otimes|u_B\rangle}
        {\sqrt{w_A^2+w_B^2}}.
\end{align*}
The real normalized product satisfies the gauge condition of
Sec.~\ref{sec:curvature}.
Here $w_B=\tfrac12\sin\theta\cos\theta$.
Then, $\|\phi_B-\phi_B^\infty\|=2\sin(\theta/4)$
and
\begin{align*}
 \max\{w_B,\|\phi_B-\phi_B^\infty\|\}
 \le\frac{a_N}{2\zeta_Ns}.
\end{align*}
Next, consider $w_A$.
Since $z_0=0$ gives $r_0=0$, Eq.~\eqref{eq:stair-tangent}, Eq.~\eqref{eq:stair-proj}
with $|x\rangle=|r\rangle$ and Eq.~\eqref{eq:stair-ybounds} give $w_A\ge\beta_N\|r\|/(1+b^2)$.
With $\|r\|\ge(1-3b^2)\|f\|$ from Lemma~\ref{lem:stair-unitcompare}, we have
\begin{align*}
 w_A(a_k)\ge\frac{\beta_N(1-3b^2)}{1+b^2}\|f(a_k)\|
 >\frac{\beta_N}{5}.
\end{align*}
The last inequality comes from $\|f(a_k)\|\ge2/9$.

At $s=a_k$ with $k\le\lfloor N/2\rfloor$, we have $a_k\ge2^{-N/2}$, so $\zeta_N>1/2$
bounds $w_B$ and $\|\phi_B-\phi_B^\infty\|$ by $2^{-N/2}$. Apply
$\|\psi_1/\|\psi_1\|-\psi_2/\|\psi_2\|\|\le2\|\psi_1-\psi_2\|/\|\psi_2\|$ with
$|\psi_1\rangle=w_A\,|u_A\rangle\otimes|\phi_B\rangle+w_B\,|\phi_A\rangle\otimes|u_B\rangle$ and
$|\psi_2\rangle=w_A\,|u_A\rangle\otimes|\phi_B\rangle$. This gives
$\|u-u_A\otimes\phi_B\|\le2w_B/w_A$. By the triangle inequality,
\begin{align*}
 \|u-u_A\otimes\phi_B^\infty\|
 &\le\frac{2w_B}{w_A}+\|\phi_B-\phi_B^\infty\|\\
 &\le\left(\frac{10\sqrt N}{b}+1\right)2^{-N/2}<b.
\end{align*}
For $N\ge 2^{24}$, $(10\sqrt{N}/b+1) 2^{-N/2}<b$.
Lemma~\ref{lem:stair-unitcompare} and the triangle inequality now give
Eq.~\eqref{eq:stair-productcompare}.
\end{proof}

\begin{thm}\label{thm:stair-curvature}
Let $N\ge2^{24}$ be an integer and let $H_N(s)$ be defined by
Eq.~\eqref{eq:stair-H} for $s\in[0,1]$.
Let $K_P$ be its total curvature. Then
\begin{align}
 \frac{N}{1600}\leq K_P<8(N+2).
 \label{eq:stair-LK}
\end{align}
Along the family $\{H_N\}$,
$K_P=\Theta(\log\Dstar^{-1})$ as $N\to\infty$.
\end{thm}
\begin{proof}
We use $u$ and $\phi_B^\infty$ of Lemma~\ref{lem:stair-productcompare} and $e$ of
Lemma~\ref{lem:stair-unitcompare}. Theorem~\ref{thm:stair-length} gives the bounds
for $L_P$, and $L_P>0$ gives the ground state of $H_N$ a positive speed on $[0,1]$
by Lemma~\ref{lem:neverstops}.
Lemmas~\ref{lem:stair-fourshift} and~\ref{lem:stair-productcompare} give
for $1\le k$ and $k+4\le\lfloor N/2\rfloor$
\begin{align*}
 \|u(a_k)-u(a_{k+4})\|>\frac1{30}-14b>\frac1{100}.
\end{align*}
The intervals $[a_{k+4},a_k]$ with $k=1,5,9,\ldots$ and
$k+4\le\lfloor N/2\rfloor$ have disjoint interiors.
There are at least $N/8-2$ such intervals. Summing $\|u(a_k)-u(a_{k+4})\|$
over them,
\begin{align}
 \int_0^1\|u'\|\,ds\ge\frac{N}{800}-\frac1{50}.
 \label{eq:stair-turnvar}
\end{align}

Moreover, with $Q=\mathbb I-|\Phi_0\rangle\langle\Phi_0|$, differentiation of
$\langle\Phi_0|u\rangle=0$ gives $u'=Qu'-v_P|\Phi_0\rangle$.
Hence Eq.~\eqref{eq:KPdef} and Eq.~\eqref{eq:stair-turnvar} imply
\begin{align}
 K_P&\ge\int_0^1\|u'\|\,ds-L_P\nonumber\\
 &>\frac{N}{800}-\frac1{50}-\sqrt{2(N+2)}
 \ge\frac{N}{1600}.
 \nonumber
\end{align}
The last inequality comes from the fact that $\sqrt{2(N+2)}<N/2000$ and $1/50\le N/8000$ for $N\ge2^{24}$.
Then Theorem~\ref{thm:K} with Lemma~\ref{lem:stair-gap} and
Eq.~\eqref{eq:stair-gapint} gives $K_P\le2\Omega(N+2)<8(N+2)$, which completes
Eq.~\eqref{eq:stair-LK}. With $\Dstar=2^{-N}$, Eq.~\eqref{eq:stair-LK} gives
$K_P=\Theta(\log\Dstar^{-1})$.
\end{proof}

\subsection{\texorpdfstring{Width class $p\ge2$}{Width class p >= 2}}

Let $p\ge2$ be an integer.
Define

\begin{align}
 \beta_p&=\frac{(2p-3)!!}{2^p},\qquad
 \gamma_p=\frac{\beta_p}{(p+1)^{p-1/2}},\nonumber\\
 \eta_p&=\min\left\{\frac14,(2\beta_p)^{-1/(2p)}\right\},\qquad
 c_p=\sqrt{\frac{\eta_p}{288}}.
 \label{eq:fixedp-data}
\end{align}
For an integer $N\ge3$ let $\delta=[\eta_p/(288N^2)]^{p/(p-1)}$,
$a=\delta^{1/(2p-1)}$, and $x=2s-1$.  With $\sigma_x,\sigma_z$ the Pauli
matrices, define one two-level factor for each $1\le j\le p$ and collect them
by the parity of $j$ into an odd and an even branch,
\begin{align}
 V_j(x)&=\binom pj(a\sigma_z+\sqrt j\,x\sigma_x),\nonumber\\
 B_{p,N}(x)&=\mathop{\boxplus}_{\substack{1\le j\le p\\j\ \mathrm{odd}}}V_j(x)
 \ \oplus\
 \left[(\delta-a)\mathbb I+
 \mathop{\boxplus}_{\substack{1\le j\le p\\j\ \mathrm{even}}}V_j(x)\right].
 \nonumber
\end{align}
We set $\epsilon=8\delta$ and
\begin{align}
 \lambda_1&=\lambda_N=2,\qquad
 \lambda_{\rm int}=\frac1{N-2},\nonumber\\
 \lambda_j&=\lambda_{\rm int}\quad(2\leq j\leq N-1),\nonumber\\
 m_0&=\frac52,\qquad m_j=m_{j-1}-\lambda_j,\nonumber\\
 t_j&=\left(j-\frac{N+1}{2}\right)\frac{32\epsilon}{\lambda_{\rm int}},\nonumber\\
 q_0&=0,\qquad q_j=q_{j-1}+\lambda_jt_j,\nonumber\\
 A_{p,N}(x)&=\operatorname{diag}(q_j+m_jx)_{j=0}^N\nonumber\\
 &\quad+\epsilon\sum_{j=1}^N
 (|j-1\rangle\langle j|+|j\rangle\langle j-1|),\nonumber\\
 H_{p,N}(s)&=A_{p,N}(2s-1)\boxplus B_{p,N}(2s-1).
 \label{eq:fixedp-H}
\end{align}
As with $A_N$ and $B_N$ in Eq.~\eqref{eq:stair-H}, the growth of the length $L_P$ comes from $A_{p,N}$ and the
gap from $B_{p,N}$. The ground energies of the two branches of $B_{p,N}$
differ by $\delta$ plus an alternating binomial sum over $j$, a $p$-th finite
difference that cancels every power of $|x|$ below $|x|^{2p}$. The gap
therefore rises from its minimum $\delta$ at $x=0$ as about
$\delta+\beta_p|x|^{2p}/\delta$, and it stays of order $\delta$ for $|x|$ up to
about $\delta^{1/p}$. In $A_{p,N}$ the diagonal slopes $m_j$ decrease by
$\lambda_j$, so consecutive diagonal entries $q_{j-1}+m_{j-1}x$ and $q_j+m_jx$
cross once each at $x=t_j$. The coupling $\epsilon$ opens each crossing into an
avoided crossing. Across the crossing at $t_j$ the ground state of $A_{p,N}$ turns
from near the coordinate vector $|j-1\rangle$ to near $|j\rangle$. These two
vectors are orthogonal, so the ground state travels a distance of order one at
each crossing. By Lemma~\ref{lem:tensorsum} the gap of $H_{p,N}$ is the smaller of
the gaps of $A_{p,N}$ and $B_{p,N}$. The gap of $A_{p,N}$ is at least
$\epsilon=8\delta$ and stays of order $\epsilon$ over the interval that holds the
crossings. Beyond the outermost avoided crossings it grows
faster than the gap of $B_{p,N}$. Consecutive crossings are a distance of order $\delta N$
apart, so the $N$ of them occupy an interval of length of order $\delta N^2$. The
definition of $\delta$ makes this length of order $\delta^{1/p}$. The whole interval
then lies where the gap of $B_{p,N}$ is still about $\delta$, below $\epsilon$, so
$B_{p,N}$ sets the gap of $H_{p,N}$. The path length then adds up to
about $N$, and $\delta N^2\propto\delta^{1/p}$ gives $N\propto\delta^{-(p-1)/(2p)}$,
that is $L_P=\Theta(\Dstar^{-(p-1)/(2p)})$.

\begin{lem}\label{lem:fixedp-profile}
Let $p\ge2$ and $N\ge3$ be integers. For $z\ge0$ let
\begin{align*}
 h_p(z)=\sum_{j=0}^p(-1)^{j+1}\binom pj\sqrt{1+jz^2}.
\end{align*}
Then
\begin{align}
 \gamma_pz^{2p}&\le h_p(z)\le\beta_pz^{2p}
 &&(0\le z\le1),\nonumber\\
 \gamma_pz&<h_p(z) &&(z\ge1),\nonumber\\
 0&<h_p(z)<z &&(z>0).
 \label{eq:fixedp-profile}
\end{align}
Moreover, for $|x|\le1$ the factor $B_{p,N}(x)$ has a nondegenerate ground state, and its gap
$\Delta_B(x)$ obeys
\begin{align}
 \delta\le\Delta_B(x)
 =\delta+ah_p(|x|/a)\le\delta+|x|<2
 \qquad(|x|\le1).
 \label{eq:fixedp-gap}
\end{align}
\end{lem}
\begin{proof}
Let $g(t)=\sqrt{1+tz^2}$. Since $g(t+1)-g(t)=\int_0^1g'(t+y)\,dy$, applying this
$p$ times gives
\begin{align*}
 \sum_{j=0}^p(-1)^{p-j}\binom pj\,g(j)
 =\int_{[0,1]^p}g^{(p)}(y_1+\cdots+y_p)\,dy.
\end{align*}
The left-hand side is $(-1)^{p+1}h_p(z)$. With
$g^{(p)}(t)=(-1)^{p+1}\beta_pz^{2p}(1+tz^2)^{1/2-p}$, this gives
\begin{align*}
 h_p(z)=\beta_pz^{2p}\int_{[0,1]^p}
 \bigl(1+(y_1+\cdots+y_p)z^2\bigr)^{1/2-p}\,dy.
\end{align*}
For $z>0$ the integrand decreases strictly in each $y_k$. It equals $1$ at
$y=(0,\dots,0)$ and $(1+pz^2)^{1/2-p}$ at $y=(1,\dots,1)$, and lies strictly between
these values elsewhere in $[0,1]^p$. Hence
$\beta_pz^{2p}(1+pz^2)^{1/2-p}<h_p(z)\le\beta_pz^{2p}$ for $z>0$. The right inequality
is the upper bound in Eq.~\eqref{eq:fixedp-profile}, and the left inequality gives
$h_p(z)>0$. With $1+pz^2\le(p+1)\max\{1,z^2\}$ the left inequality also gives the two lower
bounds in Eq.~\eqref{eq:fixedp-profile}.
The identity
\begin{align*}
 \sqrt t=\frac1{2\sqrt\pi}\int_0^\infty
 u^{-3/2}(1-e^{-tu})\,du
\end{align*}
follows from the substitution $u\to u/t$ and one integration by parts. Applied to each
term of $h_p$, it gives, for $z>0$,
\begin{align*}
 h_p(z)
 &=\frac1{2\sqrt\pi}\int_0^\infty
 u^{-3/2}e^{-u}(1-e^{-z^2u})^p\,du\\
 &<\frac1{2\sqrt\pi}\int_0^\infty
 u^{-3/2}(1-e^{-z^2u})\,du=z.
\end{align*}

For the gap calculation let $z=|x|/a$.
Each $V_j(x)$ has eigenvalues $\pm\binom pj\sqrt{a^2+jx^2}$.
Applying Lemma~\ref{lem:tensorsum} factor by factor gives each branch ground energy as the sum
of the lower eigenvalues of its factors. The
ground energy of the even branch therefore exceeds that of the odd branch by
\begin{align*}
 \delta-a+\sum_{j=1}^p(-1)^{j+1}\binom pj\sqrt{a^2+jx^2}=\delta+ah_p(z).
\end{align*}
Since $\eta_p\le1/4$ gives $\delta<1$, we have $0<\delta<a=\delta^{1/(2p-1)}<1$. Each branch has a unique ground state, and
every excitation within a branch costs at least $2\sqrt{a^2+x^2}>a+|x|>\delta+|x|$.
Since $h_p(0)=0$, Eq.~\eqref{eq:fixedp-profile} gives $0\le ah_p(z)\le|x|$. The even-branch ground
energy therefore lies above the odd-branch ground energy by $\delta+ah_p(z)$, which is
at least $\delta$ and at most $\delta+|x|$, and so below every excitation within a branch.
Hence the ground state of $B_{p,N}(x)$ is the unique ground state of the odd
branch, and $\Delta_B(x)=\delta+ah_p(|x|/a)$. With $\delta<1$ and $|x|\le1$ this
proves Eq.~\eqref{eq:fixedp-gap}.
\end{proof}

\begin{lem}\label{lem:fixedp-gap}
Let $p\ge2$ and $N\ge3$ be integers, and let $\Delta_A(x)$ and $\Delta_{p,N}(s)$
be the gaps of $A_{p,N}(x)$ and $H_{p,N}(s)$.
For $1\le j\le N$ let $J_j=[t_j-8\epsilon/\lambda_j,t_j+8\epsilon/\lambda_j]$ be the
window around the crossing at $t_j$, and let $r=\eta_p\delta^{1/p}$. The windows are
disjoint and contained in $(-4r/9,4r/9)\subset(-1,1)$.
$A_{p,N}(x)$ and $H_{p,N}(s)$ have nondegenerate ground states and satisfy
\begin{align}
 \Delta_A(x)&>\Delta_B(x)\qquad(|x|\le1),\nonumber\\
 \Delta_{p,N}(s)&=\Delta_B(2s-1)\ge\delta=\Dstar.
 \nonumber
\end{align}
The gap of $B_{p,N}$ also obeys
\begin{align}
 \Delta_B(x)\ge\gamma_p|x|^p\qquad(|x|\le1).
 \label{eq:fixedp-width}
\end{align}
In particular $\Delta_{p,N}(s)\ge2^p\gamma_p|s-1/2|^p$ for $s\in[0,1]$.
\end{lem}
\begin{proof}
Write $d_k(x)=q_k+m_kx$ for the diagonal entries of $A_{p,N}(x)$, $0\le k\le N$.
For $1\le k\le N$, Eqs.~\eqref{eq:fixedp-data} and~\eqref{eq:fixedp-H} and the
definitions of $\delta$, $a$ and $r$ give
\begin{align*}
 &d_k-d_{k-1}=\lambda_k(t_k-x),\qquad
 \lambda_{\rm int}\le\lambda_k\le2=\lambda_1=\lambda_N,\\
 &t_k-t_{k-1}=\frac{32\epsilon}{\lambda_{\rm int}}\quad(k\ge2),\\
 &\max_j|t_j|=16\epsilon(N-1)(N-2),\qquad\epsilon=8\delta,\\
 &\eta_p\le\tfrac14,\qquad\delta<1,\qquad r=288\delta N^2=36\epsilon N^2,\\
 &a^{2p-1}=\delta,\qquad\beta_p\eta_p^{2p}\le\tfrac12,\qquad r=\eta_p\delta^{1/p}<a.
\end{align*}
Since $2\cdot8\epsilon/\lambda_{\rm int}<32\epsilon/\lambda_{\rm int}$, neighbouring
windows are disjoint, and
\begin{align*}
 |x|&\le16\epsilon(N-1)(N-2)+8\epsilon(N-2)\\
 &<16\epsilon N^2=\frac{4r}9\qquad\Big(x\in\bigcup_jJ_j\Big),
\end{align*}
and $4r/9<1$ because $r=\eta_p\delta^{1/p}<1/4$.

Let $x\in J_j$. Then $t_i-x\ge24\epsilon/\lambda_{\rm int}$ for $i\ge j+1$ and
$x-t_i\ge24\epsilon/\lambda_{\rm int}$ for $i\le j-1$. Summing the diagonal
differences and using $\lambda_i\ge\lambda_{\rm int}$ gives, for every $\ell\ge1$
for which the index on the left lies in $0,\dots,N$,
\begin{equation}
 \begin{split}
 d_{j+\ell}-d_j&=\sum_{i=j+1}^{j+\ell}\lambda_i(t_i-x)\ge24\epsilon,\\
 d_{j-1-\ell}-d_{j-1}&=\sum_{i=j-\ell}^{j-1}\lambda_i(x-t_i)\ge24\epsilon.
 \end{split}
 \label{eq:fixedp-separation}
\end{equation}

Let $d_{(0)}\le d_{(1)}$ be the two smallest diagonal entries of $A_{p,N}(x)$ and
$e_0\le e_1$ its two lowest eigenvalues. Write $A_{p,N}(x)=D+O$ with $D$ its
diagonal part and $O$ its off-diagonal part. Each row of $O$ has at most two nonzero
entries, both equal to $\epsilon$, so $\|O\|\le2\epsilon$.
With $|c\rangle$ the coordinate vector of $d_{(0)}$ and Weyl's inequality for $D+O$,
$e_0\le\langle c|A_{p,N}|c\rangle=d_{(0)}$ and $e_1\ge d_{(1)}-\|O\|$. The $d_k$
decrease in $k$ while $t_k<x$ and increase afterwards, so
$\{d_{(0)},d_{(1)}\}=\{d_{k-1},d_k\}$ for some $k$ and
\begin{align}
 \Delta_A\ge d_{(1)}-d_{(0)}-2\epsilon=\lambda_k|t_k-x|-2\epsilon.
 \label{eq:fixedp-DeltaA}
\end{align}
For $x\notin J_k$ the definition of $J_k$ gives $\lambda_k|t_k-x|>8\epsilon$, so
$\Delta_A\ge6\epsilon$ outside the windows.
For $x\in J_j$ split $A_{p,N}=(M\oplus C)+W$ with $M$ the block on coordinates
$j-1$ and $j$ and $C$ the block on the other coordinates. Explicitly,
\begin{align*}
 M=\begin{pmatrix}d_{j-1}&\epsilon\\ \epsilon&d_j\end{pmatrix}.
\end{align*}
The couplings between the two blocks are
\begin{align*}
 W=\epsilon\bigl(|j-2\rangle\langle j-1|+|j\rangle\langle j+1|+\mathrm{h.c.}\bigr)
\end{align*}
for $2\le j\le N-1$, while $W=\epsilon(|1\rangle\langle2|+\mathrm{h.c.})$ for $j=1$ and
$W=\epsilon(|N-2\rangle\langle N-1|+\mathrm{h.c.})$ for $j=N$. The eigenvalues of
$M$ are
\begin{align*}
 \alpha_\pm=\frac{d_{j-1}+d_j}2\pm\sqrt{\frac{(d_j-d_{j-1})^2}4+\epsilon^2}.
\end{align*}
By Eq.~\eqref{eq:fixedp-separation}, $d_{(0)}=\min\{d_{j-1},d_j\}$, and the
definition of $J_j$ gives $|d_j-d_{j-1}|=\lambda_j|t_j-x|\le8\epsilon$. Hence
\begin{align*}
 \alpha_+&\le d_{(0)}+|d_j-d_{j-1}|+\epsilon\le d_{(0)}+9\epsilon,\\
 \alpha_+-\alpha_-&\ge2\epsilon.
\end{align*}
The diagonal entries of $C$ are at least $d_{(0)}+24\epsilon$ by
Eq.~\eqref{eq:fixedp-separation}, and its off-diagonal part is a block of $O$.
Weyl's inequality therefore bounds the lowest eigenvalue of $C$ from below by
\begin{align*}
 d_{(0)}+24\epsilon-\|O\|\ge d_{(0)}+22\epsilon>d_{(0)}+9\epsilon\ge\alpha_+,
\end{align*}
so the two lowest eigenvalues of $M\oplus C$ are $\alpha_-$ and $\alpha_+$. With
$|c'\rangle$ the lowest eigenvector of $M$ extended by zeros and Weyl's inequality
for $(M\oplus C)+W$,
\begin{align*}
 e_0\le\langle c'|A_{p,N}|c'\rangle=\alpha_-,\qquad
 e_1\ge\alpha_+-\|W\|=\alpha_+-\epsilon,
\end{align*}
so $\Delta_A\ge\alpha_+-\alpha_--\epsilon\ge\epsilon$ on each window. Hence
$\Delta_A\ge\epsilon$ for $|x|\le1$.

For $|x|\le r$ we have $|x|/a<1$, and Eq.~\eqref{eq:fixedp-profile} gives
\begin{align*}
 \Delta_B-\delta&=ah_p(|x|/a)\le a\beta_p\Big(\frac{|x|}a\Big)^{2p}\\
 &=\frac{\beta_p|x|^{2p}}{\delta}\le\beta_p\eta_p^{2p}\delta\le\frac\delta2,
\end{align*}
so $\Delta_B\le3\delta/2<\epsilon\le\Delta_A$. For $|x|\ge r>\max_j|t_j|$ the index
$k$ in Eq.~\eqref{eq:fixedp-DeltaA} is $1$ or $N$, where $\lambda_k=2$ and
$|t_k-x|=|x|-\max_j|t_j|$. With $\Delta_B\le\delta+|x|$ from
Eq.~\eqref{eq:fixedp-gap}, $2\max_j|t_j|<32\epsilon N^2=8r/9$ and $\epsilon=8\delta$,
\begin{align*}
 \Delta_A-\Delta_B&\ge|x|-2\max_j|t_j|-2\epsilon-\delta\\
 &>\frac r9-17\delta=(32N^2-17)\delta>0.
\end{align*}
Hence $\Delta_A>\Delta_B$ for $|x|\le1$. The bound $\Delta_A\ge\epsilon$ and
Lemma~\ref{lem:fixedp-profile} give $A_{p,N}$ and $B_{p,N}$ nondegenerate ground
states, respectively. Lemma~\ref{lem:tensorsum} then gives $H_{p,N}$ a nondegenerate ground state
with gap $\Delta_{p,N}(s)=\Delta_B(2s-1)$, whose minimum $\delta$ is at $s=1/2$.

Equation~\eqref{eq:fixedp-profile} at $z=1$ gives $0<\gamma_p<1$.
For $|x|\le a$, it gives
\begin{align*}
 \Delta_B(x)&\ge\delta+\gamma_p|x|^{2p}/\delta\\
 &\ge2\sqrt{\gamma_p}|x|^p\ge\gamma_p|x|^p.
\end{align*}
For $a\le|x|\le1$, it gives
$\Delta_B(x)>\gamma_p|x|\ge\gamma_p|x|^p$.
This proves Eq.~\eqref{eq:fixedp-width}, and $x=2s-1$ with
$\Delta_{p,N}(s)=\Delta_B(2s-1)$ gives $\Delta_{p,N}(s)\ge2^p\gamma_p|s-1/2|^p$.
\end{proof}

\begin{lem}\label{lem:fixedp-windows}
Let $p\ge2$ and $N\ge3$ be integers, and let $J_j$ be the windows of
Lemma~\ref{lem:fixedp-gap}. Choose a real normalized and continuously differentiable ground state $|\Phi_0^A(x)\rangle$ of
$A_{p,N}(x)$ on $[-1,1]$.
Then
\begin{align}
 \int_{J_j}\|\partial_x\Phi_0^A\|\,dx>1
 \qquad(1\le j\le N).
 \label{eq:fixedp-windows}
\end{align}
\end{lem}
\begin{proof}
Lemma~\ref{lem:fixedp-gap} gives a nondegenerate ground state and
Eq.~\eqref{eq:fixedp-separation}. Let $x_\pm=t_j\pm8\epsilon/\lambda_j$ be the
endpoints of $J_j$. By Eq.~\eqref{eq:fixedp-separation} the smallest diagonal entry
of $A_{p,N}$ at an endpoint is $d_{j-1}$ or $d_j$, and
$d_j-d_{j-1}=\lambda_j(t_j-x_\pm)=\mp8\epsilon$ decides which. Let $k$ be its index,
$k=j-1$ at $x_-$ and $k=j$ at $x_+$, and let $|k\rangle$ be the coordinate vector of
that entry. Then
\begin{align*}
 d_i-d_k\ge8\epsilon\qquad(i\ne k).
\end{align*}
Write $|\Phi_0^A\rangle=v|k\rangle+|w\rangle$ with $|w\rangle\perp|k\rangle$, and
split
\begin{align*}
 A_{p,N}=d_k|k\rangle\langle k|+|b\rangle\langle k|+|k\rangle\langle b|+C_k,
\end{align*}
where $|b\rangle=\sum_{i\ne k}\langle i|A_{p,N}|k\rangle|i\rangle$ has at most two
entries $\epsilon$, so $\|b\|\le\sqrt2\,\epsilon$, and $C_k=\sum_{i,i'\ne k}\langle i|A_{p,N}|i'\rangle|i\rangle\langle i'|$.
Let $e_0\le\langle k|A_{p,N}|k\rangle=d_k$ be the lowest eigenvalue of $A_{p,N}$.
The components of $A_{p,N}|\Phi_0^A\rangle=e_0|\Phi_0^A\rangle$ orthogonal to
$|k\rangle$ give
\begin{align*}
 (C_k-e_0)|w\rangle=-v|b\rangle.
\end{align*}
For $i\neq k$, the diagonal entries of $C_k$ are $d_i\ge d_k+8\epsilon$, and its off-diagonal
entries have at most two values $\epsilon$ per row, so
$\langle w|C_k|w\rangle\ge(d_k+6\epsilon)\|w\|^2\ge(e_0+6\epsilon)\|w\|^2$.
The inner product of the displayed equation with $|w\rangle$ then gives
\begin{align*}
 6\epsilon\|w\|^2\le|v|\,|\langle w|b\rangle|\le|v|\,\|w\|\,\|b\|\le\sqrt2\,\epsilon|v|\,\|w\|.
\end{align*}
Hence $v\ne0$, since $v=0$ would force $|w\rangle=0$ and $|\Phi_0^A\rangle=0$. The Fubini--Study distance
$\vartheta_k=\arccos|v|$ between the ground state and $|k\rangle$ therefore
satisfies
\begin{align*}
 \vartheta_k\le\tan\vartheta_k=\frac{\|w\|}{|v|}\le\frac{\sqrt2\,\epsilon}{6\epsilon}<\frac14.
\end{align*}
The Fubini--Study distance between $|j-1\rangle$ and $|j\rangle$ is $\pi/2$. By the
triangle inequality the ground states at $x_-$ and $x_+$ are at distance greater than
$\pi/2-1/2>1$, and the integral in Eq.~\eqref{eq:fixedp-windows} is at least this
distance.
\end{proof}

\begin{thm}\label{thm:fixedp-length}
For integers $p\geq2$ and $N\geq3$ with $s\in[0,1]$, $H_{p,N}$ in
Eq.~\eqref{eq:fixedp-H} with constants from Eq.~\eqref{eq:fixedp-data} is real symmetric and affine as in Eq.~\eqref{eq:Hs}.
It has dimension $(2^{\lceil p/2\rceil}+2^{\lfloor p/2\rfloor})(N+1)$ and a
nondegenerate ground state. It is in width class $p$ with
\begin{align}
 \Dstar=\delta,\qquad\Gamma=2,\qquad C_W=(2/\gamma_p)^{1/p}.
 \nonumber
\end{align}
It satisfies $\|H_{p,N}'\|<2^{p+1}\sqrt p+5$ and
\begin{align}
 L_P>c_p\Dstar^{-(p-1)/(2p)}.
 \nonumber
\end{align}
At fixed $p$ the $H_{p,N}$ form a family of width class $p$ in
Definition~\ref{def:family}, along which
$L_P=\Theta(\Dstar^{-(p-1)/(2p)})$ as $N\to\infty$.
\end{thm}
\begin{proof}
Lemma~\ref{lem:fixedp-gap} gives a nondegenerate ground state and $\Dstar=\delta$.
The odd and even branches of $B_{p,N}$ have dimensions $2^{\lceil p/2\rceil}$ and
$2^{\lfloor p/2\rfloor}$, and $A_{p,N}$ has dimension $N+1$. The defining matrices
are real symmetric and affine in $s$.

The slopes $m_j$ run from $5/2$ to $-5/2$, so $\|\partial_xA_{p,N}\|=5/2$, and
$\|\partial_xV_j\|=\binom pj\sqrt j$. With $\partial_s=2\partial_x$ and the
triangle inequality for tensor sums,
\begin{align}
 \|H_{p,N}'\|&\le2\Big(\frac52+\sum_{j=1}^p\binom pj\sqrt j\Big)
 \le2\Big(\frac52+\sqrt p\sum_{j=1}^p\binom pj\Big)\nonumber\\
 &=5+2(2^p-1)\sqrt p<2^{p+1}\sqrt p+5,
 \label{eq:fixedp-norm}
\end{align}
using $\sqrt j\le\sqrt p$ and $\sum_{j=1}^p\binom pj=2^p-1$.

Eq.~\eqref{eq:fixedp-width} is the hypothesis $\Delta(s)\ge2\omega|s-s_*|^p$ of
Lemma~\ref{lem:minorant} with $s_*=1/2$ and $2\omega=2^p\gamma_p$. With $\Gamma=2$
from Eq.~\eqref{eq:fixedp-gap}, Lemma~\ref{lem:minorant} gives
$C_W=(2/\gamma_p)^{1/p}$.

For the ground state $|\Phi_0^A\rangle$ of Lemma~\ref{lem:fixedp-windows},
Lemma~\ref{lem:tensorsum} gives $v_P\ge\|\partial_s\Phi_0^A\|$. With $ds=dx/2$ and
the disjoint windows of Lemma~\ref{lem:fixedp-gap},
\begin{align*}
 L_P&\ge\int_{-1}^1\|\partial_x\Phi_0^A\|\,dx
 \ge\sum_{j=1}^N\int_{J_j}\|\partial_x\Phi_0^A\|\,dx\\
 &>N=c_p\Dstar^{-(p-1)/(2p)},
\end{align*}
where the last equality is the definition of $\delta$ solved for $N$.

At fixed $p$ the constants $\Gamma$ and $C_W$ and the bound~\eqref{eq:fixedp-norm}
do not depend on $N$, as Definition~\ref{def:family} requires.
Equation~\eqref{eq:fixedp-gap} gives the bound $\Delta\le\Gamma$ that
Lemma~\ref{lem:layer} uses.
Here $\gamma_c=\Gamma C_W^{-p}=\gamma_p$, and $\delta<\gamma_p$ once $N$ is large.
Theorem~\ref{thm:rates} supplies the matching upper bound, proving
$L_P=\Theta(\Dstar^{-(p-1)/(2p)})$ along the family.
\end{proof}

\section{Unbounded total curvature for non-affine Hamiltonians}
\label{app:winding}

Theorem~\ref{thm:K} bounds the total curvature for affine $H(s)$. Here, we
consider non-affine $H(s)$ that is twice continuously differentiable and has a
nondegenerate ground state, so $g=1$ and $K_P=K$. The ground energy $E_0(s)$ is
nondegenerate with gap $\Delta(s)=E_1(s)-E_0(s)>0$, the phase of
$|\Phi_0(s)\rangle$ is fixed by $\langle\Phi_0|\Phi_0'\rangle=0$, and
$Q=\mathbb I-|\Phi_0\rangle\langle\Phi_0|$. We assume $v_P(s)=\|\Phi_0'\|>0$
on $[0,1]$, so the unit tangent $u=\partial_s|\Phi_0\rangle/v_P$ and $K=\int_0^1\|Qu'\|\,ds$
are defined. For affine $H(s)$, Lemma~\ref{lem:neverstops} derives this
from nonconstancy of the state. Its proof uses a constant $H'$, so here
positivity of $v_P$ is an assumption. As in Theorem~\ref{thm:K}, $\Omega(s)$
is the difference between the largest and the smallest eigenvalues of
$H'(s)$, now varying with $s$.

The following family keeps the gap constant and every norm of $H'$ and
$H''$ bounded while the total curvature diverges.
\begin{prop}\label{prop:winding}
For $0<c\le1$ and an integer $M\ge2\pi c$, let $H(s)=H_1(s)\oplus H_2(s)$ on
$\mathbb C^2\oplus\mathbb C$ with
\begin{align*}
 H_1(s)&=n(s)\cdot\boldsymbol\sigma,\qquad H_2(s)=5+s,\\
 n(s)&=(\sin\alpha\cos2\pi Ms,\,\sin\alpha\sin2\pi Ms,\,\cos\alpha),
\end{align*}
where $\boldsymbol\sigma$ is the vector of Pauli matrices and $\alpha\in(0,\pi/2)$ is
fixed by $\sin\alpha=cM^{-2}$. Then $H(s)$ is analytic and has a nondegenerate ground
state, constant gap $\Delta=2$ and $v_P(s)>0$. It satisfies $\sup_s\|H'(s)\|=1$, $\Omega(s)\le2$,
$\Lambda_1\le1$ and $\mathcal{B}_1=\int_0^1\|H''\|\,ds\le4\pi^2$. Its total curvature is
$K=2\pi M\cos\alpha$, which is unbounded as $M\to\infty$.
\end{prop}
\begin{proof}
$H_1(s)$ has eigenvalues $\mp1$ and $H_2(s)\in[5,6]$, so the ground state of $H(s)$
is the ground state of $H_1(s)$. It is nondegenerate and $\Delta=2$. Here
$H'=H_1'\oplus1$ with $H_1'=n'\cdot\boldsymbol\sigma$ and
$\|n'\|=2\pi M\sin\alpha=2\pi c/M$, so the eigenvalues of $H'$ are
$\{\pm2\pi c/M,1\}$. The condition $M\ge2\pi c$ makes $1$ the largest of them, so
$\|H'(s)\|=1$, $\Omega(s)=1+2\pi c/M\le2$, and
Eq.~\eqref{eq:nonaffinebudgets} gives $\Lambda_1=(1+2\pi c/M)/2\le1$. Similarly
$\|H''\|=\|n''\|=4\pi^2c$, and $\mathcal{B}_1=\int_0^1\|H''\|\,ds\le4\pi^2$ because $c\le1$.

The ground state of $H_1(s)$ is represented on the Bloch sphere by the unit vector
$-n(s)$. The Fubini--Study length is half the Bloch-sphere angle, so these ground
states lie on a sphere of radius $1/2$. In the gauge $\langle\Phi_0|\Phi_0'\rangle=0$
the tangent $u$ is orthogonal to $|\Phi_0\rangle$, so $\|Qu'\|\,ds=\kappa\,dl$ on
that sphere as in Sec.~\ref{sec:curvature}. Here $l$ is the arc length of the curve
traced by $-n(s)/2$ and $\kappa$ is its geodesic curvature, the curvature measured
within the sphere. The vector $-n(s)$ stays at angle $\alpha$ from $(0,0,-1)$. As
$s$ runs over $[0,1]$ it traces the circle of points at that angle $M$ times, at
constant speed $v_P=\|n'\|/2=\pi c/M>0$. On a sphere of radius $1/2$, this circle
has length $\pi\sin\alpha$ and geodesic curvature $2\cot\alpha$, so
\begin{align*}
  K=M\cdot\pi\sin\alpha\cdot2\cot\alpha=2\pi M\cos\alpha,
\end{align*}
which is unbounded as $M\to\infty$ at fixed $c$, since $\cos\alpha\to1$.
\end{proof}

\end{document}